\documentclass[a4paper,10pt]{article}
\usepackage{amssymb}
\usepackage[dvips]{color}
\usepackage{graphicx}
\usepackage{amsmath}
\usepackage{a4wide}

\newtheorem{theorem}{Theorem}

\newtheorem{corollary}[theorem]{Corollary}

\newtheorem{definition}[theorem]{Definition}

\newtheorem{lemma}[theorem]{Lemma}
\newtheorem{notation}[theorem]{Notation}

\newtheorem{proposition}[theorem]{Proposition}
\newtheorem{remark}[theorem]{Remark}

\newenvironment{proof}[1][Proof]{\textbf{#1.} }{\ \rule{0.5em}{0.5em}}

\newcommand{\G} {\ensuremath{\mathcal{G}}}
\newcommand{\F} {\ensuremath{\mathcal{F}}}
\newcommand{\Prob} {\ensuremath{\mathbb{P}}}

\newcommand{\R} {\ensuremath{\mathbb{R}}}

\begin{document}

\title{Complete duality for quasiconvex dynamic risk measures on modules of
the $L^{p}$-type}
\author{Marco Frittelli\thanks{%
Dipartimento di Matematica, Universit\`{a} degli Studi di Milano, marco.frittelli@unimi.it.} \and %
Marco Maggis\thanks{%
Dipartimento di Matematica Universit\`{a} degli Studi di Milano, marco.maggis@unimi.it.}}
\maketitle

\begin{abstract}
In the conditional setting we provide a complete duality between quasiconvex risk
measures defined on $L^{0}$ modules of the $L^{p}$ type and the appropriate class of dual functions. This is based  on a general result which extends the usual Penot-Volle representation for quasiconvex real valued maps.  
\end{abstract}



\section{Introduction}

The already-fifteen-years-old theory of risk measures is still originating
many questions and springing out lots of new problems which trigger off the interest of researchers. Recently Kupper and Schachermayer \cite{KS10}
showed that in a dynamic framework only the entropic risk measure is in
agreement with all the usual assumptions such as cash additivity,
monotonicity, convexity, law invariance and time consistency. It's thus
natural to question if these assumptions are too restrictive and indeed cash additivity was the first to be doubted and weakened to cash subadditivity, by El Karoui and Ravanelli \cite{ER09}.

Currently a debate between convexity and quasiconvexity is trying to give a better explanation to the concept of diversification, see Cerreia-Vioglio,
Maccheroni, Marinacci and Montrucchio \cite{CMMMa}. On one hand
quasiconvexity can be considered as the mathematical translation of the
principle of diversification; on the other, under the cash additivity
assumption, convexity and quasiconvexity are equivalent. Once we give up
cash additivity we are automatically induced to enlarge the class of
feasible risk measures to the class of quasiconvex funtionals. In \cite{FMP12} the authors show that on the level of distributions there do not exist any convex lower semicontinuous risk measure, but they provide a huge class of quasiconvex lower semicontinuous risk measures which contains as particular cases  the  Value at Risk, the Worst Case risk measure and the Certainty Equivalents. In \cite%
{CMMMa}, the representation of a quasiconvex cash subadditive (real valued)
risk measure $\rho $ is written in terms of the quasiaffine dual function $R$. 

\bigskip

A \textit{complete duality} for real valued quasiconvex functionals has been firstly established in \cite{CMMMb}:
the idea is to prove a \textit{one to one} relationship between quasiconvex
monotone functionals $\rho $ and the function $R$ in the dual
representation. Obviously $R$ will be unique only in an opportune class of
maps satisfying certain properties. In Decision Theory the function $R$ can
be interpreted as the decision maker's index of uncertainty aversion: the uniquesness of $R$ becomes crucial (see \cite{CMMMb} and \cite{KD}) if we want to guarantee a robust dual representation of $\rho$ characterized in terms of the unique $R$. 

\bigskip

In the conditional setting, where the maps take values in a set of random variables - for example, $\rho :L^{p}(\Omega ,\mathcal{F}_{T},\mathbb{P}%
)\rightarrow L^{p}(\Omega ,\mathcal{F}_{t},\mathbb{P})$, $t<T$ - the
representation of dynamic quasiconvex maps is obtained adopting a similar function $R$ (see \cite{FM09}). The particular case of the Conditional Certainty
Equivalent is treated in Frittelli Maggis \cite{FM10}. 
We stress that this framework is very relevant in all applications involving dynamic features and as far as we know a complete duality in this framework was lacking in literature.

\bigskip

As described in \cite{KD} topological vector spaces are the utmost general environment in which we are naturally led to embed the theory of risk preferences in the static case. On the other hand   once we shift the problem to the conditional case (as in \cite{FM09}) the mathematical challenges become harder and harder so that topological vector spaces appears as unsuitable structures.  Recently Filipovic, Kupper and Vogelpoth \cite{FKV10} discussed many advantages of working in a
module framework whenever dealing with the conditional setting. The
intuition behind the use of modules is simple and natural: given a
probability space $(\Omega ,\mathcal{F}_{T},\mathbb{P})$ and a filtration $%
\mathcal{F}=\left\{ \mathcal{F}_{t}\right\} _{0\leq t\leq T}$ , suppose that
a set $L$ of time-$T$ maturity contingent claims is fixed (for concreteness
let $L=L^{p}(\mathcal{F}_{T})$) and an agent is computing the risk of a
portfolio at an intermediate time $t<T$. All the $\mathcal{F}_{t}$%
-measurable random variables are going to be known at time $t$, thus the $%
\mathcal{F}_{t}$ measurable random variables will act as scalars in the
process of diversification of our portfolio, forcing to consider the new set 
\begin{eqnarray*}
L_{\mathcal{F}_{t}}^{p}(\mathcal{F}_{T}) &:=&L^{0}(\Omega ,\mathcal{F}_{t},%
\mathbb{P})\cdot L^{p}(\Omega ,\mathcal{F}_{T},\mathbb{P}) \\
&=&\left\{ YX\mid Y\in L^{0}(\Omega ,\mathcal{F}_{t},\mathbb{P})\text{, }%
X\in L^{p}(\Omega ,\mathcal{F}_{T},\mathbb{P})\right\}
\end{eqnarray*}%
as the domain of the risk measures. This product structure is exactly the
one that constitutes the nature of $L^{0}$-modules. %
The most significant contribution on this topic comes from the extensive research produced by Guo from 1992 until today. An useful reference is \cite{Guo1} on which the most important results are resumed and compared with the present literature and in particular with the recent
development provided by Filipovic Kupper and Vogelpoth \cite{FKV09}. The key
point in both \cite{FKV09} and \cite{Guo} is to provide a conditional form
of the Hyperplane Separation Theorems. It is well known that many
fundamental results in Mathematical Finance rely on these: for instance
Arbitrage Theory and the duality results on risk measure or utility
maximization.

In \cite{FKV09} and \cite{Guo} the authors brilliantly succeed in the task
of giving a  topological structure to $L^{0}$-modules and
to extend those theorems from functional analysis, which are relevant for
financial applications. Once this rigorous analytical background has been
carefully built up, it is possible to develop it further and obtain other
interesting results and applications. 

\bigskip

The overall aim of this paper is  the establishment of a complete duality for evenly quasiconvex \emph{conditional} risk measures (Theorem \ref{RM}). Our findings may be adapted in a \textit{dynamic framework} in Decision Theory (see \cite{CMMMb}). As explained in Section \ref{CondRM} evenly quasiconvexity of the map $\rho$ is an assumption weaker than lower semicontinuity and quasiconvexity.
\\As already mentioned,  uniqueness of the representation is a delicate issue in the conditional case: once embedded in the the of $L^0$-modules  the complete duality for conditional risk measures (see Theorem \ref{RM} for the precise statement), perfectly matches what had been obtained in \cite{CMMMb} for the static case and provide great evidences of the power of the module approach.

Let $\G\subset \F$ be two sigma algebras we deduce under suitable conditions that $\rho :L_{\mathcal{G}}^{p}(\mathcal{F}%
)\rightarrow L^{0}(\mathcal{G})$ is an evenly quasiconvex conditional risk
measure \emph{if and only if} 
\begin{equation}
\rho (X)=\sup_{Q\in \mathcal{P}^{q}}R\left( E_Q\left[ -X|%
\mathcal{G}\right] ,Q\right)  \label{introRM}
\end{equation}%
where $\mathcal{P}^{q}$ is a subset of probabilities $Q$ such that $E[\frac{dQ}{d\Prob}|\G]=1$ and $R$ is \emph{unique} in the class $%
\mathcal{M}(L^{0}(\mathcal{G})\times \mathcal{P}^{q})$ (see the Definition %
\ref{KK}). In particular $R$ will take the form 
\begin{equation*}
R(Y,Q)=\inf_{\xi \in L_{\mathcal{G}}^{p}(\mathcal{F})}\left\{ \rho (\xi
)\mid E_{Q}\left[ -\xi  | \mathcal{G}\right] =Y\right\}
\end{equation*}
\emph{A posteriori} if we add the assumption $\rho(X+\alpha)=\rho(X)-\alpha$ for every $\alpha\in L^0(\G)$, then the quasiconvex map $\rho$ is automatically convex and $R(Y,Q)=Y-\rho^{*}(-Q)$ (see Corollary \ref{corCASmod}) so that we recover the dual representation proved in \cite{Sca}.
\\  The function $R$ has also an interesting interpretation related to the dual representation of convex risk measures.  It's not hard to show that for every $X\in L_{\mathcal{G%
}}^{p}(\mathcal{F}),$ $Q\in \mathcal{P}^{q}$ and any map $\rho :L_{\mathcal{G%
}}^{p}(\mathcal{F})\rightarrow L^{0}(\mathcal{G})$ we have: 
\begin{equation}
R\left( E_Q\left[ -X|\mathcal{G}\right] ,Q\right) \geq E_Q
\left[ -X|\mathcal{G}\right] -\rho ^{\ast }(-Q),
\label{ddd}
\end{equation}
where $\rho^{\ast}$ is the convex conjugate of $\rho$.
\\From equation (\ref{introRM}) we deduce that whenever the preferences of an agent are described by a quasiconvex - not convex - risk measure we cannot recover the risk only taking a \emph{supremum} of the Fenchel conjugate, i.e. of the RHS of (\ref{ddd}), over all the possible probabilistic
scenarios. We shall need a more cautious approach represented by the new penalty function $R\left( E_Q\left[ -X|\mathcal{G}\right],Q\right) $. The quantity $R(Y,Q)$ is therefore the reserve amount required at the intermediate time $t$ ($\mathcal{F}_{t}=\mathcal{G}$) under the scenario $Q$, to cover an expected loss $Y\in L^{0}(\mathcal{G})$ in the
future.

\bigskip

The paper is organized as follows. In Section \ref{prel} we provide some preliminary notions and facts: a short review about $L^{0}$-modules of the $L^p$ type and the concept of conditionally evenly convex set. Section \ref{CondRM} is devoted to the regularity, quasiconvexity and continuity assumptions of the maps $\rho :L^p_{\G}(\F)\rightarrow L^{0}(\mathcal{G})$. In Section \ref{CompSect} we state the complete duality for quasiconvex conditional risk measures. We include in
Section \ref{complements} some complementary results. Section \ref{conti} is devoted to the proofs of the main contributions of the paper. Two more technical lemmas are deferred to the Appendix.


\section{Notations, setting and topological properties}\label{prel}

The probability space $(\Omega ,\mathcal{F},\mathbb{P})$ is fixed throughout
this chapter and $\mathcal{G\subseteq F}$ is any sigma algebra contained in $%
\mathcal{F}$. We denote with $L^{0}(\Omega ,\mathcal{F},\mathbb{P})=L^{0}(%
\mathcal{F})$ (resp. $L^{0}(\mathcal{G})$ ) the space of $\mathcal{F}$
(resp. $\mathcal{G}$) measurable random variables that are $\mathbb{P}$ a.s.
finite, whereas by $\bar{L}^{0}(\mathcal{F})$ the space of extended random
variables which may take values in $\mathbb{R}\cup \{\infty \}$. In general
since $(\Omega ,\mathbb{P})$ are fixed we will always omit them. We define $%
L_{+}^{0}(\mathcal{F})=\{X\in L^{0}(\mathcal{F})\mid X\geq 0\}$ and $%
L_{++}^{0}(\mathcal{F})=\{X\in L^{0}(\mathcal{F})\mid X>0\}$. We remind that
all equalities/inequalities among random variables are meant to hold $%
\mathbb{P}$-a.s.. As the expected value $E_{\mathbb{P}}[\cdot ]$ is mostly
computed w.r.t. the reference probability $\mathbb{P}$, we will often omit $%
\mathbb{P}$ in the notation. \newline
Moreover the essential ($\mathbb{P}$ almost surely) \emph{supremum} $%
ess\sup_{\lambda }(X_{\lambda })$ of an arbitrary family of random variables 
$X_{\lambda }\in L^{0}(\Omega ,\mathcal{F},\mathbb{P})$ will be simply
denoted by $\sup_{\lambda }(X_{\lambda })$, and similarly for the essential 
\emph{infimum}. The symbol $\vee $ (resp. $\wedge $) denotes the essential ($%
\mathbb{P}$ almost surely) \emph{maximum} (resp. the essential \emph{minimum}%
) between two random variables, which are the usual lattice operations.

\bigskip

\paragraph{On $L^0$ modules of the $L^p$ type.} We now introduce the structure of normed module of the $L^p$ type which play a key role in the financial applications and are studied in detail in \cite{KV09} Section 4.2.

\bigskip

Consider the generalized conditional expectation of $%
\mathcal{F}$-measurable non negative random variables: $E[\cdot |\mathcal{G}%
]:L_{+}^{0}(\mathcal{F})\rightarrow \bar{L}_{+}^{0}(\mathcal{G})$ 
\begin{equation*}
E[X|\mathcal{G}]=:\lim_{n\rightarrow +\infty }E[X\wedge n|\mathcal{G}].
\end{equation*}%
The basic properties of conditional expectation still hold true. In
particular for every $X,X_{1},X_{2}\in L_{+}^{0}(\mathcal{F})$ and $Y\in
L^{0}(\mathcal{G})$

\begin{description}
\item[(i)] $YE[X|\mathcal{G}]=E[YX|\mathcal{G}]$;

\item[(ii)] $E[X_{1}+X_{2}|\mathcal{G}]=E[X_{1}|\mathcal{G}]+E[X_{2}|\mathcal{%
G}]$;

\item[(iii)] $E[X]=E[E[X|\mathcal{G}]]$.
\end{description}

$L^{0}(\mathcal{G})$ equipped with the order of the almost sure dominance is
a lattice ordered ring.  Let $p\in \lbrack 1,\infty ]$ and consider the algebraic module over the ring $L^{0}(\G)$ defined as 
\begin{equation*}
L_{\mathcal{G}}^{p}(\mathcal{F})=:\{X\in L^{0}(\Omega ,\mathcal{F},\mathbb{P}%
)\mid \Vert X|\mathcal{G}\Vert _{p}\in L^{0}(\Omega ,\mathcal{G},\mathbb{P}%
)\}
\end{equation*}%
where $\Vert \cdot |\mathcal{G}\Vert _{p}$ is 
assigned by 
\begin{equation}
\Vert X|\mathcal{G}\Vert _{p}=:\left\{ 
\begin{array}{cc}
E[|X|^{p}|\mathcal{G}]^{\frac{1}{p}} & \text{ if }p<+\infty \\ 
\inf \{Y\in \bar{L}^{0}(\mathcal{G})\mid Y\geq |X|\} & \text{ if }p=+\infty%
\end{array}%
\right.  \label{norm}
\end{equation}
By this definition $L_{\mathcal{G}}^{p}(\mathcal{F})$ inherits the product structure i.e. 
\begin{equation*}
L_{\mathcal{G}}^{p}(\mathcal{F})=L^{0}(\mathcal{G})L^{p}(\mathcal{F}%
)=\{YX\mid Y\in L^{0}(\mathcal{G}),\;X\in L^{p}(\mathcal{F})\}.
\end{equation*}
This last property allows the conditional expectation to be well defined for
every $\widetilde{X}\in L_{\mathcal{G}}^{p}(\mathcal{F})$; indeed, if $%
\widetilde{X}=YX$ with $Y\in L^{0}(\mathcal{G})$ and $X\in L^{p}(\mathcal{F}%
),$ then $E[\widetilde{X}|\mathcal{G}]=YE[X|\mathcal{G}]$ is a finite valued
random variable.
Moreover $\Vert \cdot |\mathcal{G}\Vert _{p}$ is a  $L^{0}(\mathcal{G})$-norm according to the following definition.
\begin{definition}
A map $\|\cdot\|: L^p_{\G}(\F) \rightarrow L^0_+$ is a $L^0$-norm on $L^p_{\G}(\F)$ if

\begin{description}
\item[(i)] $\|\gamma X\|=|\gamma|\|X\|$ for all $\gamma\in L^0$ and $X\in L^p_{\G}(\F)$,

\item[(ii)] $\|X_1+X_2\|\leq \|X_1\|+\|X_2\|$ for all $X_1,X_2\in L^p_{\G}(\F)$.

\item[(iii)] $\Vert X \Vert =0$ implies $X=0$.

\end{description}
\end{definition}

If we endow $L^0(\G)$ with a topology $\tau_0$ we may automatically induce a topology $\tau$ on $L^p_{\G}(\F)$ by 
$$ X_{\alpha}\overset{\tau}{\rightarrow} X\quad \text{ if and only if } \quad \|X_{\alpha}-X\|_p \overset{\tau_0}{\rightarrow} 0 $$
Two natural choices for $\tau_0$ are the topology of the convergence in probability (as used in \cite{Guo}) or the uniform topology as introduced in  \cite{FKV09}. In the following Remark we recall the second one since is non-standard in the literature.

\begin{remark}\label{uniform} For every $\varepsilon \in L_{++}^{0}(\mathcal{G}),$
the ball $B_{\varepsilon }:=\{Y\in L^{0}({\mathcal{G}})\mid |Y|\leq
\varepsilon \}$ centered in $0\in L^{0}(\mathcal{G})$ gives the neighborhood
basis of $0$. A set $V\subset L^{0}(\mathcal{G})$ is a neighborhood of $Y\in
L^{0}(\mathcal{G})$ if there exists $\varepsilon \in L_{++}^{0}(\mathcal{G})$
such that $Y+B_{\varepsilon }\subset V$. A set $V$ is open if it is a
neighborhood of all $Y\in V$. $(L^{0}(\mathcal{G}),|\cdot |)$ stands for $%
L^{0}(\mathcal{G})$ endowed with this topology: in this case the space
looses the property of being a topological vector space. It is easy to see
that a net converges in this topology, namely $Y_{N}\overset{|\cdot |}{%
\rightarrow }Y$ if for every $\varepsilon \in L_{++}^{0}(\mathcal{G})$ there
exists $\overline{N}$ such that $|Y-Y_{N}|<\varepsilon $ for every $N>%
\overline{N}$.
\end{remark}

Given the pair $(L^p_{\G}(\F),\tau)$,  $(L^0(\G),\tau_0)$ the dual module of $(L^p_{\G}(\F))^{\ast}$
will be the collection of continuous functional $\mu :(L_{\mathcal{G}}^{p}(
\mathcal{F}),\tau)\rightarrow (L^{0}(\mathcal{G}),\tau_0)$ which are $L^0(\G)$-linear i.e.

$$\mu(\alpha X+ \beta Y)= \alpha \mu(X)+\beta \mu(Y)$$
for every $\alpha,\beta\in L^0(\G)$ and $X,Y\in L^p_{\G}(\F)$.
\\From now on we will consider on $L^0(\G)$ any topology $\tau_0$ such that the dual module can be identified with a random
variable 
$$Z\in L_{\mathcal{G}}^{q}(\mathcal{F})  \rightleftarrows  \mu (\cdot )=E[Z\cdot |\mathcal{G}]$$ 
where $\frac{1}{p}+\frac{1}{q}=1$. So in general we can
identify the dual modules of $L^p_{\G}(\F)$  with $L^q_{\mathcal{G} }(\mathcal{F} )$. This occurs in particular if $\tau_0$ is the one defined in Remark \ref{uniform} or  the topology of convergence in probability (see \cite{Guo} and \cite{FKV09}).

\paragraph{On conditionally evenly convex sets.} We recall that a subset $V$ of a locally convex topological vector space is \emph{evenly convex} if it is the intersection of a family of open half spaces, or equivalently, if every $X\notin V$ can be separated from $V$ by a continuous linear functional. Obviously an evenly convex set is necessarily convex and moreover the whole topological vector space is a trivial case of evenly convex set. In this subsection we recall the generalization of the concept of  evenly convexity as introduced in \cite{FM12}, which is tailor made for the conditional setting.  We refer the reader to \cite{FM12} for further details and motivations.

\begin{definition}  Let $\mathcal{C}$ be a subset of $L^p_{\G}(\F)$. 
\begin{description}
\item[(CSet)] $\mathcal{C}$ has the countable
concatenation property if for every countable partition $\{A_{n}\}_{n}%
\subseteq \mathcal{G}$ and for every countable collection of elements $%
\{X_n\}_n\subset \mathcal{C}$ we have $\sum_{n}\mathbf{1}_{A_{n}}X_n\in 
\mathcal{C}$.
\end{description}
We denote by $\mathcal{C}^{cc}$ the countable
concatenation hull of $\mathcal{C}$, namely the smallest set $\mathcal{C}%
^{cc}\supset \mathcal{C}$ which satisfies (CSet).
\end{definition}

We notice that an arbitrary set $\mathcal{C}\subset L_{\mathcal{G}}^{p}(\mathcal{F}) $ may present some components which degenerate to
the entire module. Basically it might occur that for some $A\in \mathcal{G}$
, $\mathcal{C}\mathbf{1}_{A}=L_{\mathcal{G}}^{p}(\mathcal{F}) \mathbf{1}_{A}$, i.e., for each $\xi \in L_{\mathcal{G}}^{p}(\mathcal{F}) $
there exists $\eta \in \mathcal{C}$ such that $\eta \mathbf{1}_{A}=\xi 
\mathbf{1}_{A}$. In this case there are no chances to guarantee a separation on the set $\Omega$ as for the results given in \cite{FKV09}. Thus we need to determine the maximal $\mathcal{G}$-measurable set on which $\mathcal{C}$ reduces to $L_{\mathcal{G}}^{p}(\mathcal{F}) $. The existence of the maximal element has been proved in \cite{FM12} (see Remark \ref{remMAX} in the Appendix of the present paper) and the following
definition is well posed.

\begin{notation}
\label{trivialcomponent}Fix a set $\mathcal{C}\subseteq L_{\mathcal{G}}^{p}(\mathcal{F}) $. As the class $%
\mathcal{A}(\mathcal{C}):=\{A\in \mathcal{G}\mid \mathcal{C}\mathbf{1}_{A}=L_{\mathcal{G}}^{p}(\mathcal{F}) 
\mathbf{1}_{A}\}$ is closed with respect to countable union, we denote with $A_{\mathcal{C}}$ the $\mathcal{G}$-measurable maximal element of the class $\mathcal{A}(\mathcal{C})$ and with $D_{\mathcal{C}}$ the (P-a.s. unique) complement of $A_{\mathcal{C}}$ (see also the Remark \ref{remMAX}). Hence $\mathcal{C}\mathbf{1}_{A_{\mathcal{C}}}=L_{\mathcal{G}}^{p}(\mathcal{F}) \mathbf{1}_{A_{\mathcal{C}}}.$
\end{notation}

\begin{definition} Let $\mathcal{C}$ be a subset of $L^p_{\G}(\F)$. We will say that
 \begin{description}
\item[(i)]  $X \in L^p_{\G}(\F)$ \text{ is outside } $\mathcal{C}$ if $\mathbf{1}_{A}\{X\}\cap \mathbf{1}_{A}\mathcal{C} =\varnothing$  for every  $A\in 
\mathcal{G}$  with  $A\subseteq D_{\mathcal{C}}$  and $\mathbb{P}(A)>0$.
\item[(ii)]  $\mathcal{C} $ is  $L^0$-\emph{convex} if $\Lambda X_1+(1-\Lambda)X_2\in \mathcal{C}$ for every $X_1,X_2\in\mathcal{C}$, $\Lambda\in L^0(\G)$ and $0\leq Lambda \leq 1$.
 \item[(iii)]  $\mathcal{C} $ is \emph{conditional evenly convex} if $\mathcal{C} $ satisfies (CSet) and for every $X$ outside $\mathcal{C}$ there exists a $\mu \in L_{\mathcal{G}}^{q}(\mathcal{F}) $ such that 
\begin{equation*}
\mu (X)>\mu (\xi )\text{ on } D_{\mathcal{C}},\;\forall \,\xi \in \mathcal{C}%
.
\end{equation*}
\end{description}

\end{definition}

In \cite{FM12} it is showed that any conditional evenly convex set is also $L^0$-convex and it can be characterized as intersection of half spaces.

\section{Conditional Risk Measures defined on $L_{%
\mathcal{G}}^{p}(\mathcal{F})$.}\label{CondRM}  

In this section we summarize the properties of the risk maps that we will consider  and recall the Penot-Volle type dual representation of quasiconvex conditional maps as proved in \cite{FM12}. 

\begin{definition}
\label{defdef1}A map $\rho : L^p_{\G}(\F)\rightarrow \bar{L}^{0}(\mathcal{G})$ is
\begin{description}
\item[(REG)] regular if for every $X_{1},X_{2}\in L^p_{\G}(\F)$ and $A\in \mathcal{G}$, 
\begin{equation*}
\rho (X_{1}\mathbf{1}_{A}+X_{2}\mathbf{1}_{A^{C}})=\rho (X_{1})\mathbf{1}%
_{A}+\rho (X_{2})\mathbf{1}_{A^{C}}.
\end{equation*}
\end{description}
\end{definition}

\begin{remark}
\label{remB} It is well known that (REG) is equivalent to:%
\begin{equation*}
\rho (X\mathbf{1}_{A})\mathbf{1}_{A}=\rho (X)\mathbf{1}_{A}\text{, }\forall
A\in \mathcal{G}\text{, }\forall X\in L^p_{\G}(\F).
\end{equation*}%
In the module setting it is even true that (REG) is equivalent to countably
regularity, i.e. 
\begin{equation*}
\rho (\sum_{i=1}^{\infty }X_{i}\mathbf{1}_{A_{i}})=\sum_{i=1}^{\infty }\rho
(X_{i})\mathbf{1}_{A_{i}}\text{ on \ }\cup _{i=1}^{\infty }A_{i}
\end{equation*}
for $X_{i}\in L^p_{\G}(\F)$ and $\left\{ A_{i}\right\} _{i}$  a sequence of disjoint $\mathcal{G}$ measurable sets. Indeed, from the module properties, $X:=\sum_{i=1}^{\infty }X_{i}\mathbf{1}_{A_{i}}\in L^p_{\G}(\F)$ and $\sum_{i=1}^{\infty}\rho (X_{i})\mathbf{1}_{A_{i}}\in \bar{L}^{0}(\mathcal{G})$; (REG) then implies $\rho (X)\mathbf{1}_{A_{i}}=\rho (X\mathbf{1}_{A_{i}})\mathbf{1}_{A_{i}}=\rho (X_{i}\mathbf{1}_{A_{i}})\mathbf{1}_{A_{i}}=\rho (X_{i})\mathbf{1%
}_{A_{i}}$.

\end{remark}

Let $\rho : L^p_{\G}(\F)\rightarrow \bar{L}^{0}(\mathcal{G})$ be (REG). There might exist
a set $A\in \mathcal{G}$ on which the map $\rho$ is infinite, in the sense
that $\rho (\xi )\mathbf{1}_{A}=+\infty\mathbf{1}_{A}$ for every $\xi  \in L^p_{\G}(\F)$. For this reason we introduce 
\begin{equation*}
\mathcal{A}:=\{A\in \mathcal{G}\mid \rho (\xi )\mathbf{1}_{A}=+\infty\mathbf{1%
}_{A}\;\forall \,\xi \in L^p_{\G}(\F)\}.
\end{equation*}
Applying Lemma \ref{L10} in Appendix with $F:=\left\{ \rho (\xi )\mid \xi \in
L^p_{\G}(\F)\right\} $ and $Y_{0}=+\infty$ we can deduce the existence of two maximal
sets $T_{\rho }\in \mathcal{G}$ and $\Upsilon _{\rho }\in \mathcal{G}$ for
which $P(T_{\rho }\cap \Upsilon _{\rho })=0$, $P(T_{\rho }\cup \Upsilon _{\rho
})=1$ and 
\begin{eqnarray}
\rho (\xi )=+\infty &\text{ on }\Upsilon _{\rho }&\text{for every }\xi ,\eta
\in L^p_{\G}(\F),  \notag \\
\rho (\zeta)<+\infty &\text{ on }T_{\rho }&\text{ for some }\zeta\in L^p_{\G}(\F).
\label{888}
\end{eqnarray}%
Suppose that a map $\rho$ satisfies: $\mathbb{P}(\Upsilon _{\rho })>0$ so
that $\rho (\xi )\mathbf{1}_{\Upsilon _{\rho }}=+\infty \mathbf{1}_{\Upsilon
_{\rho }} $ for every $\xi \in L^p_{\G}(\F)$. Then its lower level sets $\{X\in L^p_{\G}(\F)\mid \rho (X)\leq \eta \}$, $\eta \in L^{0}(\mathcal{G})$, are all empty (and so convex and closed). This would imply that any such map is quasiconvex and lower semicontinuous, regardless of its behavior on (the relevant set) $T_{\rho }$. This explains the need of introducing the set $T_{\rho }$ in the
following definition of (QCO), (LSC) and (EVQ). \newline
Hereafter we state the conditional version of some relevant properties of
the maps under investigation. To this aim we define, for $Y\in L^{0}(\mathcal{G}),
$$$U_{\rho }^{Y}:=\{\xi \in L^p_{\G}(\F)\mid \rho (\xi )\mathbf{1}_{T_{\rho }}\leq Y\}.$$

\begin{definition}
\label{defdef}Let $\rho :L^p_{\G}(\F)\rightarrow \bar{L}^{0}(\mathcal{G})$. The map $\rho $ is:

\begin{description}
\item[(QCO)] quasiconvex if the sets $U_{\rho }^{Y}$ are $L^{0}(\mathcal{G})$%
-convex $\forall Y\in L^{0}(\mathcal{G})$.

\item[(EVQ)] evenly quasiconvex if the sets $U_{\rho }^{Y}$ are conditional evenly convex $\forall Y\in L^{0}(\mathcal{G})$.

\item[(LSC)] lower semicontinuous if the sets $U_{\rho
}^{Y}$ are closed $\forall Y\in L^{0}(\mathcal{G}).$

\end{description}
\end{definition}

\begin{remark}
\label{remA}Let $\rho :L^p_{\G}(\F)\rightarrow \bar{L}^{0}(\mathcal{G})$.
\begin{description}
\item[(i)] The quasiconvexity of $\rho $ is equivalent to the condition 
\begin{equation}
\rho (\Lambda X_{1}+(1-\Lambda )X_{2})\leq \rho (X_{1})\vee \rho (X_{2}),
\label{max}
\end{equation}%
for every $X_{1},X_{2}\in L^p_{\G}(\F)$, $\Lambda \in L^{0}(\mathcal{G})$ and $0\leq\Lambda \leq 1$.
\item[(ii)] If the map $\rho$ is (REG) then $U_{\rho}^{Y}$ satisfies (CSet).
\end{description}
\end{remark}

Regularity also guarantees that evenly quasiconvex maps are indeed
quasiconvex. Moreover the next Lemma shows that the property (EVQ) is weaker than (QCO) plus (LSC)  (see \cite{FM12} Section 5 for further details).

\begin{lemma}
\label{remconv}Let $\rho :L^p_{\G}(\F)\rightarrow \bar{L}^{0}(\mathcal{G})$ be (REG).
\begin{description}
\item[(i)] If $\rho $ is (EVQ) then it is (QCO).
\item[(ii)] If $\rho$ is (QCO) and (LSC)  then it is (EVQ).
\end{description}
\end{lemma}

\begin{theorem}[\cite{FM12} Theorem 16]
\label{EvQco}If $\rho :L^p_{\G}(\F)\rightarrow \bar{L}^{0}(\mathcal{G})$ is (REG) and
(EVQ) then 
\begin{equation}
\rho (X)=\sup_{\mu \in L^q_{\G}(\F)}\mathcal{R}(\mu (X),\mu ),  \label{rapprPiSup}
\end{equation}%
where for $Y\in L^{0}(\mathcal{G})$ and $\mu \in L^q_{\G}(\F)$, 
\begin{equation}
\mathcal{R}(Y,\mu ):=\inf_{\xi \in L^p_{\G}(\F)}\left\{ \rho (\xi )\mid \mu (\xi )\geq
Y\right\} .  \label{1212}
\end{equation}
\end{theorem}

\begin{definition}
We say that a map $\rho :L_{\mathcal{G}}^{p}(\mathcal{F})\rightarrow \bar{L}%
^{0}(\mathcal{G})$ is:
\end{definition}

\begin{description}
\item[($\downarrow $ MON)] \textit{monotone decreasing if} $X_{1}\geq
X_{2}\Rightarrow \rho (X_{1})\leq \rho (X_{2})$.
\end{description}
\begin{definition}
A \emph{quasiconvex conditional risk measure} is
a map $\rho :L_{\mathcal{G}}^{p}(\mathcal{F})\rightarrow \bar{L}^{0}(%
\mathcal{G})$ satisfying (REG), ($\downarrow $MON) and (QCO).
\end{definition}

Recall that the principle of diversification states that \textquotedblleft
diversification should not increase the risk\textquotedblright , i.e. the
diversified position $\Lambda X+(1-\Lambda )Y$ is less risky than either the
positions $X$ or $Y$. Thus the mathematical formulation of this priciple is
exactly quasiconvexity, i.e. the property (\ref{max}). Under the cash
additivity axiom

\begin{description}
\item[(CAS)] $\rho (X+\Lambda )=\rho (X)-\Lambda $, for any $\Lambda \in
L^{0}(\mathcal{G}$) and $X\in L_{\mathcal{G}}^{p}(\mathcal{F})$,
\end{description}

\noindent convexity and quasiconvexity are equivalent, so that they both
provide the right interpretation of this principle. As vividly discussed by
El Karoui and Ravanelli \cite{ER09} the lack of liquidity of zero coupon
bonds is the primary reason of the failure of cash additivity. In addition,
in the \emph{time consistent} case, the cash subadditivity property

\begin{description}
\item[(CSA)] $\rho (X+\Lambda )\geq \rho (X)-\Lambda $, for any $\Lambda \in
L_{+}^{0}(\mathcal{G}$) and $X\in L_{\mathcal{G}}^{p}(\mathcal{F})$,
\end{description}

\noindent is the adequate property of a conditional risk measure \emph{for
processes} (see the discussion in Section 5, \cite{FoPe}).

Thus it is unavoidable in the dynamic setting to relax the convexity axiom
to quasiconvexity (and (CAS) to (CSA)) in order to regain the best modeling
of diversification.

\subsection{Complete Duality}\label{CompSect}

This section is devoted to main result of this paper: a complete
quasiconvex conditional duality between the risk measure $\rho $ and the dual map $R$. Given any $\rho$ we guarantee the existence of a unique map $R$ in the class $\mathcal{M}(L^{0}(\mathcal{G})\times \mathcal{P}^{q})$ which allows a dual representation of the form given by equation (\ref{rapprRM}). On the other hand every $R\in \mathcal{M}(L^{0}(\mathcal{G})\times \mathcal{P}^{q})$ will identify a unique quasiconvex conditional risk measure by the mean of the representation  (\ref{rapprRM}). 

\bigskip

As discussed in the previous section, the duality concerning conditional
quasiconvex risk measures holds only under an additional continuity
assumption (either (EVQ) or (LSC)). For the analysis of the
complete duality in this section, we chose the weakest assumption, i.e. (EVQ),
and we leave the (LSC) case for further investigation. We stress that,
in virtue of Lemma \ref{remconv}, any map satisfying the assumptions of Theorem \ref{RM} is a conditional quasiconvex risk measure.

Due to the assumption that $\rho $ is monotone \textit{decreasing}, we
modify, with just a difference in the sign, the definition of the dual
function and rename it as: 

\begin{equation}
R(Y,Z):=\inf_{\xi \in L_{\mathcal{G}}^{p}(\mathcal{F})}\left\{ \rho (\xi
)\mid E\left[ -\xi Z|\mathcal{G}\right] \geq Y\right\} \text{.}  \label{RLP}
\end{equation}%
The function $R$ is well defined on the domain 
\begin{equation}
\Sigma =\{(Y,Z)\in L^{0}(\mathcal{G})\times L_{\mathcal{G}}^{q}(\mathcal{F}%
)\mid \exists \xi \in L_{\mathcal{G}}^{p}(\mathcal{F})\text{ s.t. }E[-Z\xi |%
\mathcal{G}]\geq Y\}.  \label{sigma}
\end{equation}%
Let also introduce the following set: 
\begin{equation*}
\mathcal{P}^{q}=:\{Z\in L_{\mathcal{G}}^{q}(\mathcal{F})\mid Z\geq 0,\;E[Z|%
\mathcal{G}]=1\}
\end{equation*}%
and with a slight abuse of notation we will write that the probability $Q\in 
\mathcal{P}^{q}$ instead of $Z=\frac{dQ}{d\mathbb{P}}\in \mathcal{P}^{q}$
and $R(Y,Q)$ instead of $R\left( Y,\frac{dQ}{d\mathbb{P}}\right) $. In addition for every $Q\in \mathcal{P}^q$ we have $E\left[\frac{dQ}{d\mathbb{P}}X|
\mathcal{G}\right]=E_Q\left[X|\mathcal{G}\right]$.

\begin{definition}
\label{KK}The class $\mathcal{M}(L^{0}(\mathcal{G})\times \mathcal{P}^{q})$
is composed by maps $K:L^{0}(\mathcal{G})\times \mathcal{P}^{q}\rightarrow 
\bar{L}^{0}(\mathcal{G})$ s.t.
\begin{description}
\item[(i)] $K$ is increasing in the first component.

\item[(ii)] $K(Y\mathbf{1}_{A},Q)\mathbf{1}_{A}=K(Y,Q)\mathbf{1}_{A}$ for
every $A\in \mathcal{G}$ and $(Y,\frac{dQ}{d\mathbb{P}})\in \Sigma $.

\item[(iii)] $\inf_{Y\in L^{0}(\mathcal{G})}K(Y,Q)=\inf_{Y\in L^{0}(\mathcal{G%
})}K(Y,Q^{\prime })$ for every $Q,Q^{\prime }\in \mathcal{P}^{q}$.

\item[(iv)] $K$ is $\diamond $-evenly $L^{0}(\mathcal{G})$-quasiconcave: for
every $(Y^{\ast },Q^{\ast })\in L^{0}(\mathcal{G})\times \mathcal{P}^{q}$, $%
A\in \mathcal{G}$ and $\alpha \in L^{0}(\mathcal{G})$ such that $K(Y^{\ast
},Q^{\ast })<\alpha $ on $A$, there exists $(S^{\ast },X^{\ast })\in
L_{++}^{0}(\mathcal{G})\times L_{\mathcal{G}}^{p}(\mathcal{F})$ with 
\begin{equation*}
Y^{\ast }S^{\ast }+E\left[ X^{\ast }\frac{dQ^{\ast }}{d\mathbb{P}}\mid 
\mathcal{G}\right] <YS^{\ast }+E\left[ X^{\ast }\frac{dQ}{d\mathbb{P}}\mid 
\mathcal{G}\right] \text{ on }A
\end{equation*}%
for every $(Y,Q)$ such that $K(Y,Q)\geq \alpha $ on $A$.

\item[(v)] the set $\mathcal{K}(X)=\left\{ K(E[X\frac{dQ}{d\mathbb{P}}|%
\mathcal{G}],Q)\mid Q\in \mathcal{P}^{q}\right\} $ is upward directed for
every $X\in L_{\mathcal{G}}^{p}(\mathcal{F})$.

\item[(vi)] $K(Y,Q_{1})\mathbf{1}_{A}=K(Y,Q_{2})\mathbf{1}_{A}$, if $\frac{%
dQ_{1}}{d\mathbb{P}}\mathbf{1}_{A}=\frac{dQ_{2}}{d\mathbb{P}}\mathbf{1}_{A},$
$Q_{i}\in \mathcal{P}^{q}$, and $A\in \mathcal{G}.$
\end{description}
\end{definition}

We will show in Lemma \ref{classM} that the class $\mathcal{M}(L^{0}(%
\mathcal{G})\times \mathcal{P}^{q})$ is not empty.

\begin{theorem}
\label{RM} The map $\rho :L_{\mathcal{G}}^{p}(\mathcal{F})\rightarrow L^{0}(\mathcal{G})$ satisfies (REG), ($\downarrow $MON), (EVQ) if and only if 
\begin{equation}
\rho (X)=\sup_{Q\in \mathcal{P}^{q}}R\left( E_Q\left[ -X|%
\mathcal{G}\right] ,Q\right)  \label{rapprRM}
\end{equation}%
where 
\begin{equation*}
R(Y,Q)=\inf_{\xi \in L_{\mathcal{G}}^{p}(\mathcal{F})}\left\{ \rho (\xi
)\mid E_Q\left[ -\xi |\mathcal{G}\right] =Y\right\}
\end{equation*}%
is unique in the class $\mathcal{M}(L^{0}(\mathcal{G})\times \mathcal{P}%
^{q}) $.
\end{theorem}

\begin{proof}
In Section \ref{P2}.
\end{proof}

\subsection{Complements}\label{complements}

From Theorem \ref{RM} we can deduce the following proposition which confirm what was obtained in \cite{FM09}.

\begin{proposition}
Suppose that $\rho $ satisfies the same assumptions of Theorem \ref{RM}.
Then the restriction $\widehat{\rho }:=\rho \mathbf{1}_{L^{p}(\mathcal{F})}$
defined by $\widehat{\rho }(X)=\rho (X)$ for every $X\in L^{p}(\mathcal{F})$
can be represented as 
\begin{equation*}
\widehat{\rho }(X)=\sup_{Q\in \mathcal{P}^{q}}\inf_{\xi \in L^{p}(\mathcal{F}%
)}\left\{ \widehat{\rho }(\xi )\mid E_Q[-\xi |\mathcal{G}%
]=E_Q[-X|\mathcal{G}]\right\} .
\end{equation*}
\end{proposition}

\begin{proof}
For every $X\in L^{p}(\mathcal{F})$, $Q\in \mathcal{P}^{q}$ we have 
\begin{eqnarray*}
\widehat{\rho }(X) &\geq &\inf_{\xi \in L^{p}(\mathcal{F})}\left\{ \widehat{%
\rho }(\xi )\mid E_Q[-\xi |\mathcal{G}]=E_Q[-X|\mathcal{G}]\right\} \\
&\geq &\inf_{\xi \in L_{\mathcal{G}}^{p}(\mathcal{F})}\left\{ \rho (\xi
)\mid E_Q[-\xi|\mathcal{G}]=E_Q[-X|%
\mathcal{G}]\right\}
\end{eqnarray*}%
and hence the thesis.
\end{proof}

\bigskip

The following result is meant to confirm that the dual representation chosen
for quasiconvex maps is indeed a good generalization of the convex case.

\begin{corollary}
\label{corCASmod}Let $\rho :L_{\mathcal{G}}^{p}(\mathcal{F})\rightarrow
L^{0}(\mathcal{G})$.

\noindent (i) If $Q\in \mathcal{P}^q$ and if $\rho$ is (MON), (REG) and
(CAS) then 
\begin{equation*}
R(E_{Q}(-X|\mathcal{G} ),Q)=E_{Q}(-X|\mathcal{G} )-\rho ^{\ast }(-Q)
\end{equation*}
where 
\begin{equation*}
\rho ^{\ast }(-Q)=\sup_{\xi\in L^p_{\mathcal{G} }(\mathcal{F}
)}\left\{E_{Q}[-\xi|\mathcal{G} ]-\rho(\xi)\right\} .  \label{KKK}
\end{equation*}

\noindent (ii) Under the same assumptions of Theorem \ref{RM} and if $\rho $
satisfies in addition (CAS) then 
\begin{equation*}
\rho (X)=\sup_{Q\in \mathcal{P}^{q}}\left\{ E_{Q}(-X|\mathcal{G})-\rho
^{\ast }(-Q)\right\} .
\end{equation*}
\end{corollary}

\noindent

\begin{proof}
Denote $\mu (\cdot )=:E_Q\left[ \cdot \mid \mathcal{G}%
\right] $. By definition of $R$ 
\begin{eqnarray*}
R(E_{Q}(-X|\mathcal{G}),Q) &=&\inf_{\xi \in L_{\mathcal{G}}^{p}(\mathcal{F}%
)}\left\{ \rho (\xi )\mid \mu (-\xi )=\mu (-X)\right\} \\
&=&\mu (-X)+\inf_{\xi \in L_{\mathcal{G}}^{p}(\mathcal{F})}\left\{ \rho (\xi
)-\mu (-X)\mid \mu (-\xi )=\mu (-X)\right\} \\
&=&\mu (-X)+\inf_{\xi \in L_{\mathcal{G}}^{p}(\mathcal{F})}\left\{ \rho (\xi
)-\mu (-\xi )\mid \mu (-\xi )=\mu (-X)\right\} \\
&=&\mu (-X)-\sup_{\xi \in L_{\mathcal{G}}^{p}(\mathcal{F})}\left\{ \rho (\xi
)-\mu (-X)\mid \mu (-\xi )=\mu (-X)\right\} \\
&=&\mu (-X)-\rho ^{\ast }(-Q),
\end{eqnarray*}%
where the last equality follows from 
\begin{eqnarray*}
\rho ^{\ast }(-Q)\overset{(CAS)}{=} &&\sup_{\xi \in L_{\mathcal{G}}^{p}(%
\mathcal{F})}\left\{ \mu (-\xi -\mu (X-\xi ))-\rho (\xi +\mu (X-\xi
))\right\} \\
&=&\sup_{\eta \in L_{\mathcal{G}}^{p}(\mathcal{F})}\left\{ \mu (-\eta )-\rho
(\eta )\mid \eta =\xi +\mu (X-\xi )\right\} \\
&\leq &\sup_{\eta \in L_{\mathcal{G}}^{p}(\mathcal{F})}\left\{ \mu (-\eta
)-\rho (\eta )\mid \mu (-\eta )=\mu (-X)\right\} \leq \rho ^{\ast }(-Q).
\end{eqnarray*}%
(ii) It is a consequence of (i) and Theorem \ref{RM}.
\end{proof}

\subsection{A characterization \emph{via} the risk acceptance family}

In this subsection we assume for the sake of simplicity that $\rho (0)\in
L^{0}(\mathcal{G})$ which implies that $\Prob(T_{\rho})=1$. In this way we do not loose any generality imposing $\rho (0)=0$ (if not, just define $\tilde{\rho}(\cdot ):=\rho (\cdot )-\rho
(0)$). We remind that if $\rho (0)=0$ then (REG) is equivalent to the
condition 
\begin{equation*}
\rho (X\mathbf{1}_{A})=\rho (X)\mathbf{1}_{A},\text{ }A\in L^{0}(\mathcal{G}).
\end{equation*}%
Given a risk measure one can always define for every $Y\in L^{0}(\mathcal{G}%
) $ the risk acceptance set of level $Y$ as 
\begin{equation*}
\mathcal{A}_{\rho }^{Y}=\{X\in L_{\mathcal{G}}^{p}(\mathcal{F})\mid \rho
(X)\leq Y\}.
\end{equation*}%
This set represents the collection of financial positions whose risk is
smaller of the fixed level $Y$ and are strictly related to the Acceptability
Indices \cite{CM}. Given a risk measure $\rho $ we can associate a family of
risk acceptance sets, namely $\{\mathcal{A}_{\rho }^{Y}\mid Y\in L^{0}(%
\mathcal{G})\}$, as it was suggested in the static case in \cite{KD}. 

\begin{definition}
A family $\mathbb{A}=\{\mathcal{A}^{Y}\mid Y\in L^{0}(\mathcal{G})\}$ of
subsets $\mathcal{A}^{Y}\subset L_{\mathcal{G}}^{p}(\mathcal{F})$ is called
risk acceptance family if the following properties hold:

\noindent (i) convexity: $\mathcal{A}^{Y}$ is $L^{0}(\mathcal{G})$-convex
for every $Y\in L^{0}(\mathcal{G})$;

\noindent (ii) monotonicity:

$\cdot $ $X_{1}\in \mathcal{A}^{Y}$ and $X_{2}\in L_{\mathcal{G}}^{p}(%
\mathcal{F})$, $X_{2}\geq X_{1}$ implies $X_{2}\in \mathcal{A}^{Y}$;

$\cdot $ $\mathcal{A}^{Y_{1}}\subseteq \mathcal{A}^{Y_{2}}$ for any $%
Y_{1}\leq Y_{2}$, $Y_{i}\in L^{0}(\mathcal{G});$

\noindent (iii) regularity: fix $X\in \mathcal{A}^{Y}$ then for every $G\in 
\mathcal{G}$ we have 
\begin{equation*}
\inf \{Y\mathbf{1}_{G}\mid Y\in L^{0}(\mathcal{G})\text{ s.t. }X\in \mathcal{%
A}^{Y}\}=\inf \{Y\mid Y\in L^{0}(\mathcal{G})\text{ s.t. }X\mathbf{1}_{G}\in 
\mathcal{A}^{Y}\}
\end{equation*}
\end{definition}

These three properties allow us to induce a one to one relationship between
quasiconvex conditional risk measures and risk acceptance families as we
prove in the following

\begin{proposition}
For any quasiconvex conditional risk measure $\rho :L_{\mathcal{G}}^{p}(%
\mathcal{F})\rightarrow \bar{L}^{0}(\mathcal{G})$ the family 
\begin{equation*}
\mathbb{A}_{\rho }=\{\mathcal{A}_{\rho }^{Y}\mid Y\in L^{0}(\mathcal{G})\}
\end{equation*}%
with $\mathcal{A}_{\rho }^{Y}=\{X\in L_{\mathcal{G}}^{p}(\mathcal{F})\mid
\rho (X)\leq Y\}$ is a risk acceptance family. \newline
\emph{Viceversa} for every risk acceptance family $\mathbb{A}$ the map 
\begin{equation*}
\rho _{\mathbb{A}}(X)=\inf \{Y\mid Y\in L^{0}(\mathcal{G})\text{ s.t. }X\in 
\mathcal{A}^{Y}\}
\end{equation*}%
is a well defined quasiconvex conditional risk measure $\rho _{\mathbb{A}%
}:L_{\mathcal{G}}^{p}(\mathcal{F})\rightarrow \bar{L}^{0}(\mathcal{G})$ such
that $\rho (0)=0$. \newline
Moreover, $\rho _{\mathbb{A}_{\rho }}=\rho $ and if $\mathcal{A}%
^{Y}=\bigcap_{Y^{\prime }>Y}\mathcal{A}^{Y^{\prime }}$ for every $Y\in L^{0}(%
\mathcal{G})$ then $\mathbb{A}_{\rho _{\mathbb{A}}}=\mathbb{A}$.
\end{proposition}

\begin{proof}
The proof is an extension from the static case provided in \cite{CM} and 
\cite{KD}. \newline
($\downarrow $MON) and (QCO) of $\rho $ imply that $\mathcal{A}_{\rho }^{Y}$
is convex and monotone. Also notice that 
\begin{eqnarray*}
&&\inf \{Y\mid Y\in L^{0}(\mathcal{G})\text{ s.t. }X\mathbf{1}_{G}\in 
\mathcal{A}_{\rho }^{Y}\}=\inf \{Y\mid \rho (X\mathbf{1}_{G})\leq Y\text{
for }Y\in L^{0}(\mathcal{G})\} \\
&=&\rho (X\mathbf{1}_{G})=\rho (X)\mathbf{1}_{G}=\inf \{Y\mathbf{1}_{G}\mid
Y\in L^{0}(\mathcal{G})\text{ s.t. }\rho (X)\leq Y\} \\
&=&\inf \{Y\mathbf{1}_{G}\mid Y\in L^{0}(\mathcal{G})\text{ s.t. }X\in 
\mathcal{A}_{\rho }^{Y}\},
\end{eqnarray*}%
i.e. $\mathcal{A}_{\rho }^{Y}$ is regular.

\bigskip

\emph{Viceversa:} we first prove that $\rho _{\mathbb{A}}$ is (REG). For
every $G\in \mathcal{G}$ 
\begin{eqnarray*}
\rho _{\mathbb{A}}(X\mathbf{1}_{G})= &&\inf \{Y\mid Y\in L^{0}(\mathcal{G})%
\text{ s.t. }X\mathbf{1}_{G}\in \mathcal{A}^{Y}\} \\
\overset{(iii)}{=} &&\inf \{Y\mathbf{1}_{G}\mid Y\in L^{0}(\mathcal{G})\text{
s.t. }X\in \mathcal{A}^{Y}\}=\rho _{\mathbb{A}}(X)\mathbf{1}_{G}
\end{eqnarray*}%
Now consider $X_{1},X_{2}\in L_{\mathcal{G}}^{p}(\mathcal{F})$, $X_{1}\leq
X_{2}$. Let $G^{C}=\{\rho _{\mathbb{A}}(X_{1})=+\infty \}$ so that $\rho _{%
\mathbb{A}}(X_{1}\mathbf{1}_{G^{C}})\geq \rho _{\mathbb{A}}(X_{2}\mathbf{1}%
_{G^{C}})$. Otherwise consider the collection of $Y$s such that $X_{1}%
\mathbf{1}_{G}\in \mathcal{A}^{Y}$. Since $\mathcal{A}^{Y}$ is monotone we
have that $X_{2}\mathbf{1}_{G}\in \mathcal{A}^{Y}$ if $X_{1}\mathbf{1}%
_{G}\in \mathcal{A}^{Y}$ and this implies that 
\begin{eqnarray*}
\rho _{\mathbb{A}}(X_{1})\mathbf{1}_{G} &=&\inf \{Y\mathbf{1}_{G}\mid Y\in
L^{0}(\mathcal{G})\text{ s.t. }X_{1}\in \mathcal{A}^{Y}\} \\
&=&\inf \{Y\mid Y\in L^{0}(\mathcal{G})\text{ s.t. }X_{1}\mathbf{1}_{G}\in 
\mathcal{A}^{Y}\} \\
&\geq &\inf \{Y\mid Y\in L^{0}(\mathcal{G})\text{ s.t. }X_{2}\mathbf{1}%
_{G}\in \mathcal{A}^{Y}\} \\
&=&\inf \{Y\mathbf{1}_{G}\mid Y\in L^{0}(\mathcal{G})\text{ s.t. }X_{2}\in 
\mathcal{A}^{Y}\}=\rho (X_{2})\mathbf{1}_{G},
\end{eqnarray*}%
i.e. $\rho _{\mathbb{A}}(X_{1}\mathbf{1}_{G})\geq \rho _{\mathbb{A}}(X_{2}%
\mathbf{1}_{G})$. And this shows that $\rho _{\mathbb{A}}(\cdot )$ is ($%
\downarrow $MON). \newline
Let $X_{1},X_{2}\in L_{\mathcal{G}}^{p}(\mathcal{F})$ and take any $\Lambda
\in L^{0}(\mathcal{G})$, $0\leq \Lambda \leq 1$. Define the set $B=:\{\rho _{%
\mathbb{A}}(X_{1})\leq \rho _{\mathbb{A}}(X_{2})\}$. If $X_{1}\mathbf{1}%
_{B^{C}}+X_{2}\mathbf{1}_{B}\in \mathcal{A}^{Y^{\prime }}$ for some $%
Y^{\prime }\in L^{0}(\mathcal{G})$ then for sure $Y^{\prime }\geq \rho _{%
\mathbb{A}}(X_{1})\vee \rho _{\mathbb{A}}(X_{2})\geq \rho (X_{i})$ for $%
i=1,2 $. Hence also $\rho (X_{i})\in \mathcal{A}^{Y^{\prime }}$ for $i=1,2$
and by convexity we have that $\Lambda X_{1}+(1-\Lambda )X_{2}\in \mathcal{A}%
^{Y^{\prime }}$. Then $\rho _{\mathbb{A}}(\Lambda X_{1}+(1-\Lambda
)X_{2})\leq \rho _{\mathbb{A}}(X_{1})\vee \rho _{\mathbb{A}}(X_{2})$. 
\newline
If $X_{1}\mathbf{1}_{B^{C}}+X_{2}\mathbf{1}_{B}\notin \mathcal{A}^{Y^{\prime
}}$ for every $Y^{\prime }\in L^{0}(\mathcal{G})$ then from property (iii)
we deduce that $\rho _{\mathbb{A}}(X_{1})=\rho _{\mathbb{A}}(X_{2})=+\infty $
and the thesis is trivial.

Now consider $B=\{\rho (X)=+\infty \}$: $\rho _{\mathbb{A}_{\rho }}(X)=\rho
(X)$ follows from 
\begin{eqnarray*}
\rho _{\mathbb{A}_{\rho }}(X)\mathbf{1}_{B} &=&\inf \{Y\mathbf{1}_{B}\mid
Y\in L^{0}(\mathcal{G})\text{ s.t. }\rho (X)\leq Y\}=+\infty \mathbf{1}_{B}
\\
\rho _{\mathbb{A}_{\rho }}(X)\mathbf{1}_{B^{C}} &=&\inf \{Y\mathbf{1}%
_{B^{C}}\mid Y\in L^{0}(\mathcal{G})\text{ s.t. }\rho (X)\leq Y\} \\
&=&\inf \{Y\mid Y\in L^{0}(\mathcal{G})\text{ s.t. }\rho (X)\mathbf{1}%
_{B^{C}}\leq Y\}=\rho (X)\mathbf{1}_{B^{C}}
\end{eqnarray*}%
For the second claim notice that if $X\in \mathcal{A}^{Y}$ then $\rho _{%
\mathbb{A}}(X)\leq Y$ which means that $X\in \mathcal{A}_{\rho _{\mathbb{A}%
}}^{Y}$. Conversely if $X\in \mathcal{A}_{\rho _{\mathbb{A}}}^{Y}$ then $%
\rho _{\mathbb{A}}(X)\leq Y$ and by monotonicity this implies that $X\in 
\mathcal{A}^{Y^{\prime }}$ for every $Y^{\prime }>Y$. From the right
continuity we take the intersection and get that $X\in \mathcal{A}^{Y}$.
\end{proof}

\section{Proofs}\label{conti}

\subsection{General properties of $\mathcal{R}(Y,\protect\mu )$}

Following the path traced in \cite{FM09} and \cite{FM12}, we adapt (without giving a proof) to the $L^p$ module framework
the foremost properties holding for the function $\mathcal{R}%
:L^{0}(\mathcal{G})\times L^q_{\G}(\F) \rightarrow \bar{L}^{0}(\mathcal{G})$
defined in (\ref{1212}). Let the effective domain of the function $\mathcal{R
}$ be: 
\begin{equation}
\Sigma _{\mathcal{R}}:=\{(Y,\mu )\in L^{0}(\mathcal{G})\times L^q_{\G}(\F)\mid
\exists \xi \in L^p_{\G}(\F)\text{ s.t. }\mu (\xi )\geq Y\}.  \label{domainMod}
\end{equation}

\begin{lemma}
\label{down} Let $\mu \in L^q_{\G}(\F)$, $X\in L^p_{\G}(\F)$ and $\rho :L^p_{\G}(\F)\rightarrow \bar{L}^{0}(\mathcal{G})$ satisfy (REG).

\noindent i) $\mathcal{R}(\cdot ,\mu )$ is monotone non decreasing.

\noindent ii) $\mathcal{R}(\Lambda \mu (X),\Lambda \mu )=\mathcal{R}(\mu
(X),\mu )$ for every $\Lambda \in L^{0}(\mathcal{G})$.

\noindent iii) For every $Y\in L^{0}(\mathcal{G})$ and $\mu \in L^q_{\G}(\F)$,
the set 
\begin{equation*}
\mathcal{A}_{\mu }(Y)\circeq \{\rho(\xi )\,|\,\xi \in L^p_{\G}(\F),\;\mu (\xi )\geq Y\}
\end{equation*}%
is downward directed in the sense that for every $\rho (\xi _{1}),\rho (\xi
_{2})\in \mathcal{A}_{\mu }(Y)$ there exists $\rho (\xi ^{\ast })\in \mathcal{%
A}_{\mu }(Y)$ such that $\rho (\xi ^{\ast })\leq \min \{\rho (\xi _{1}),\rho
(\xi _{2})\}$.

In addition, if $\mathcal{R}(Y,\mu )<\alpha $ for some $\alpha \in L^{0}(%
\mathcal{G})$ then there exists $\xi $ such that $\mu (\xi )\geq Y$ and $\rho
(\xi )<\alpha $.

\noindent iv) For every $A\in \mathcal{G}$, $(Y,\mu )\in \Sigma _{\mathcal{R}%
}$ 
\begin{eqnarray}
\mathcal{R}(Y,\mu )\mathbf{1}_{A} &=&\inf_{\xi \in L^p_{\G}(\F)}\left\{ \rho (\xi )%
\mathbf{1}_{A}\mid Y\geq \mu (X)\right\}  \label{TakingOutM} \\
&=&\inf_{\xi \in L^p_{\G}(\F)}\left\{ \rho(\xi )\mathbf{1}_{A}\mid Y\mathbf{1}_{A}\geq
\mu (X\mathbf{1}_{A})\right\} =\mathcal{R}(Y\mathbf{1}_{A},\mu )\mathbf{1}%
_{A}  \label{222}
\end{eqnarray}

\noindent v) For every $X_{1},X_{2}\in L^p_{\G}(\F)$

\qquad (a) $\mathcal{R}(\mu (X_{1}),\mu )\wedge \mathcal{R}(\mu (X_{2}),\mu
)=\mathcal{R}(\mu (X_{1})\wedge \mu (X_{2}),\mu )$

\qquad (b) $\mathcal{R}(\mu (X_{1}),\mu )\vee \mathcal{R}(\mu (X_{2}),\mu )=%
\mathcal{R}(\mu (X_{1})\vee \mu (X_{2}),\mu )$

\noindent vi) The map $\mathcal{R}(\mu (X),\mu )$ is quasi-affine with
respect to $X$ in the sense that for every $X_{1},X_{2}\in L^p_{\G}(\F)$, $\Lambda\,\in L^{0}(\mathcal{G})$ and $0\leq \Lambda \leq 1,$ we have

$\qquad \mathcal{R}(\mu (\Lambda X_{1}+(1-\Lambda )X_{2}),\mu )\geq \mathcal{%
R}(\mu (X_{1}),\mu )\wedge \mathcal{R}(\mu (X_{2}),\mu )\text{
(quasiconcavity)}$

$\qquad \mathcal{R}(\mu (\Lambda X_{1}+(1-\Lambda )X_{2}),\mu )\leq \mathcal{%
R}(\mu (X_{1}),\mu )\vee \mathcal{R}(\mu (X_{2}),\mu )\text{
(quasiconvexity).}$

\noindent vii) $\inf_{Y\in L^{0}(\mathcal{G})}\mathcal{R}(Y,\mu
_{1})=\inf_{Y\in L^{0}(\mathcal{G})}\mathcal{R}(Y,\mu _{2})$ for every $\mu
_{1},\mu _{2}\in L^q_{\G}(\F)$.
\end{lemma}


\subsection{\label{P2}Proofs of the complete duality stated in Section 
\protect\ref{CompSect}}

We need some preliminary results

\begin{lemma}
\label{classM} Let $\rho$ be (REG). The function $R$ defined in (\ref{RLP})
belongs to $\mathcal{M}(L^0(\mathcal{G} )\times \mathcal{P}^q)$
\end{lemma}

\begin{proof}
We check the items in Definition \ref{KK}.

\noindent\ i) and ii) can be easily shown.  \newline
iii) Observe that $R(Y,Q)\geq \inf_{\xi \in L_{\mathcal{G}}^{p}(\mathcal{F}%
)}\rho (\xi ),$ for all $(Y,Q)\in L^{0}(\mathcal{G})\times \mathcal{P}^{q},$
so that 
\begin{equation*}
\inf_{Y\in L^{0}(\mathcal{G})}R(Y,Q)\geq \inf_{\xi \in L_{\mathcal{G}}^{p}(%
\mathcal{F})}\rho (\xi ).
\end{equation*}%
Conversely notice that the set $\{\rho (\xi )\mid \xi \in L_{\mathcal{G}%
}^{p}(\mathcal{F})\}$ is downward directed and then there exists $\rho (\xi
_{n})\downarrow \inf_{\xi \in L_{\mathcal{G}}^{p}(\mathcal{F})}\rho (\xi )$.
For every $Q\in \mathcal{P}^{q}$ we have 
\begin{equation*}
\rho (\xi _{n})\geq R\left( E\left[ -\xi _{n}\frac{dQ}{d\mathbb{P}}|\mathcal{%
G}\right] ,Q\right) \geq \inf_{Y\in L^{0}(\mathcal{G})}R(Y,Q)
\end{equation*}%
and therefore 
\begin{equation*}
\inf_{Y\in L^{0}(\mathcal{G})}R(Y,Q)\leq \inf_{\xi \in L_{\mathcal{G}}^{p}(%
\mathcal{F})}\rho (\xi ).
\end{equation*}%
iv) For $\alpha \in L^{0}(\mathcal{G})$ and $A\in \mathcal{G}$ define $%
U_{\alpha }^{A}=\{(Y,Q)\in L^{0}(\mathcal{G})\times \mathcal{P}^{q}\mid
R(Y,Q)\geq \alpha \text{ on }A\}$, and suppose $\emptyset \neq U_{\alpha
}^{A}\neq L^{0}(\mathcal{G})\times \mathcal{P}^{q}$. Let $(Y^{\ast },Q^{\ast
})\in L^{0}(\mathcal{G})\times \mathcal{P}^{q}$ such that $R(Y^{\ast
},Q^{\ast })<\alpha $ on $A$.

As mentioned in Lemma \ref{down} (iii) there exists $X^{\ast }\in L_{\mathcal{G}}^{p}(%
\mathcal{F})$ such that $E[-X^{\ast }\frac{dQ^{\ast }}{d\mathbb{P}}|\mathcal{%
G}]\geq Y^{\ast }$ and $\rho (X^{\ast })<\alpha $ on $A$. Since $R(Y,Q)\geq
\alpha $ on $A$ for every $(Y,Q)\in U_{\alpha }^{A}$, we deduce that $%
E[-X^{\ast }\frac{dQ}{d\mathbb{P}}|\mathcal{G}]<Y$ on $A,$ for every $%
(Y,Q)\in U_{\alpha }$. Otherwise we could define $B=\{\omega \in A\mid
E[-X^{\ast }\frac{dQ}{d\mathbb{P}}|\mathcal{G}]\geq Y\}\subseteq A$, $%
\mathbb{P}(B)>0$ and then Lemma \ref{down} (iv) would imply $R(Y\mathbf{1}%
_{B},Q)<\alpha $ on $B$.

Finally we can conclude that for every $(Y,Q)\in U_{\alpha }^{A}$ 
\begin{equation*}
Y^{\ast }+E\left[ X^{\ast }\frac{dQ^{\ast }}{d\mathbb{P}}|\mathcal{G}\right]
\leq 0<Y+E\left[ X^{\ast }\frac{dQ}{d\mathbb{P}}|\mathcal{G}\right] \text{
on }A.
\end{equation*}%
v) Take $Q_{1},Q_{2}\in \mathcal{P}^{q}$ and define $F=\{R(E[X\frac{dQ_{1}}{d%
\mathbb{P}}|\mathcal{G}],Q_{1})\geq R(E[X\frac{dQ_{2}}{d\mathbb{P}}|\mathcal{%
G}],Q_{2})\}$ and let $\widehat{Q}$ given by 
\begin{equation*}
\frac{d\widehat{Q}}{d\mathbb{P}}:=\mathbf{1}_{F}\frac{dQ_{1}}{d\mathbb{P}}+%
\mathbf{1}_{F^{C}}\frac{dQ_{2}}{d\mathbb{P}}\in \mathcal{P}^{q}.
\end{equation*}%
It is easy to show, using an argument similar to the one in \cite{FM09},
Lemma 3.5 v) that 
\begin{equation*}
R\left( E\left[ X\frac{d\widehat{Q}}{d\mathbb{P}}|\mathcal{G}\right] ,%
\widehat{Q}\right) =R\left( E\left[ X\frac{dQ_{1}}{d\mathbb{P}}|\mathcal{G}%
\right] ,Q_{1}\right) \vee R\left( E\left[ X\frac{dQ_{2}}{d\mathbb{P}}|%
\mathcal{G}\right] ,Q_{2}\right) ,
\end{equation*}%
which shows that the set $\left\{ R(E[X\frac{dQ}{d\mathbb{P}}|\mathcal{G}%
],Q)\mid Q\in \mathcal{P}^{q}\right\} $ is upward directed.

vi) It follows from the same argument used in \cite{FM09}, Lemma 3.5 iv).
\end{proof}

\begin{lemma}
\label{monK} If $Q\in \mathcal{P}^{q}$ and if $\rho $ is ($\downarrow $ MON)
and (REG) then 
\begin{equation}
R\left( Y,Q\right) =\inf_{\xi \in L_{\mathcal{G}}^{p}(\mathcal{F})}\left\{
\rho (\xi )\mid E\left[ -\xi \frac{dQ}{d\mathbb{P}}|\mathcal{G}\right]
=Y\right\} .  \label{pen1}
\end{equation}
\end{lemma}

\begin{proof}
For sake of simplicity we denote by $\mu (\cdot )=E[\cdot \frac{dQ}{d\mathbb{%
P}}|\mathcal{G}]$ and $r(Y,\mu )$ the right hand side of equation (\ref{pen1}%
). Notice that $R(Y,\mu )\leq r(Y,\mu )$. By contradiction, suppose that $%
\mathbb{P}(A)>0$ where $A=:\{R(Y,\mu )<r(Y,\mu )\}$. From Lemma \ref{down},
there exists a r.v. $\xi \in L_{\mathcal{G}}^{p}(\mathcal{F})$ satisfying
the following conditions

\begin{itemize}
\item $\mu(-\xi) \geq Y$ and $\mathbb{P} (\mu(-\xi) > Y)>0$.

\item $R(Y,\mu)(\omega )\leq \rho (\xi )(\omega )<r(Y,\mu)(\omega )$ for $%
\mathbb{P}$-almost every $\omega\in A$.
\end{itemize}

Then $Z:=\mu (-\xi )-Y\in L^{0}(\mathcal{G})\subseteq L_{\mathcal{G}}^{p}(%
\mathcal{F})$ satisfies $Z\geq 0$, $\mathbb{P}(Z>0)>0$ and, thanks to ($%
\downarrow $ MON), $\rho (\xi )\geq \rho (\xi +Z)$. From $\mu (-(\xi +Z))=Y$
we deduce: 
\begin{equation*}
R(Y,\mu )(\omega )\leq \rho (\xi )(\omega )<r(Y,\mu )(\omega )\leq \rho (\xi
+Z)(\omega )\text{ for }\mathbb{P}\text{-a.e. }\omega \in A,
\end{equation*}%
which is a contradiction.
\end{proof}

\subsubsection{Proof of Theorem \protect\ref{RM}}

During the whole proof we fix an arbitrary $X\in L_{\mathcal{G}}^{p}(%
\mathcal{F})$.

\paragraph{ONLY IF.}

For the proof of the \textquoteleft only if \textquoteright we here repeat for sake of completeness some arguments used in \cite{FM12}. 
\\There might exist a set $A\in \mathcal{G}$ on which the map $\rho $ is constant, in the sense that $\rho (\xi )\mathbf{1}_{A}=\rho (\eta )\mathbf{1}_{A}$ for every $\xi ,\eta \in E$. For this reason we introduce 
\begin{equation*}
\mathcal{A}:=\{B\in \mathcal{G}\mid \rho (\xi )\mathbf{1}_{B}=\rho (\eta )%
\mathbf{1}_{B}\;\forall \,\xi ,\eta \in L^p_{\G}(\F)\}.
\end{equation*}%
Applying Lemma \ref{L10} in Appendix with $F:=\left\{ \rho (\xi )-\rho (\eta
)\mid \xi ,\eta \in L^p_{\G}(\F)\right\} $ (we consider the convention $+\infty -\infty=0$) and $Y_{0}=0$ we can deduce the existence of two maximal sets $A\in \mathcal{G}$ and $A^{\vdash}\in \mathcal{G}$ for which $P(A\cap A^{\vdash})=0$, $P(A\cup A^{\vdash})=1$ and 
\begin{eqnarray}
\rho (\xi )=\rho (\eta ) &\text{ on }A&\text{for every }\xi ,\eta \in L^p_{\G}(\F), 
\notag \\
\rho (\zeta _{1})<\rho (\zeta _{2}) &\text{ on }A^{\vdash}&\text{ for some }%
\zeta _{1},\zeta _{2}\in L^p_{\G}(\F).  \label{888}
\end{eqnarray}%
Recall that $\Upsilon_{\rho}\in \mathcal{G} $ is the maximal set on which $%
\rho(\xi)\mathbf{1}_{\Upsilon_{\rho}}=+\infty\mathbf{1}_{\Upsilon_{\rho}}$ for
every $\xi\in L^p_{\G}(\F)$ and $T_{\rho}$ its complement. Notice that $%
\Upsilon_{\rho}\subset A$. \newline
Fix $X\in L^p_{\G}(\F)$ and $G=\{\rho (X)<+\infty \}$. For every $\varepsilon \in
L_{++}^{0}(\mathcal{G})$ we set 
\begin{equation}  \label{Yeps}
Y_{\varepsilon }=:0\mathbf{1}_{\Upsilon_{\rho}}+\rho(X)\mathbf{1}_{A\setminus
\Upsilon_{\rho}}+(\rho (X)-\varepsilon )\mathbf{1}_{G\cap
A^{\vdash}}+\varepsilon \mathbf{1}_{G^{C}\cap A^{\vdash}}
\end{equation}
and for every $\varepsilon\in L^0(\mathcal{G} )_{++}$ we set 
the evenly convex set 
\begin{equation*}
\mathcal{C}_{\varepsilon }=:\{\xi \in L_{\mathcal{G}}^{p}(\mathcal{F})\mid
\rho (\xi )\mathbf{1}_{T_{\rho}}\leq Y_{\varepsilon }\}\neq \emptyset
\end{equation*}
This may be separated from $X$ by $\mu _{\varepsilon }\in L^q_{\mathcal{G%
} }(\mathcal{F} )$ i.e. 
\begin{equation}
\mu _{\varepsilon }(X)>\mu _{\varepsilon }(\xi )\quad \text{ on }D_{\mathcal{%
C}_{\varepsilon}},\; \forall \,\xi \in \mathcal{C}_{\varepsilon }.
\label{HahnBan}
\end{equation}%
Let $\eta \in L_{\mathcal{G}}^{p}(\mathcal{F})$, $\eta \geq 0$. If $\xi \in 
\mathcal{C}_{\varepsilon }$ then ($\downarrow $ MON) implies $\xi +n\eta \in 
\mathcal{C}_{\varepsilon }$ for every $n\in \mathbb{N}$. Since $\mu
_{\varepsilon }(\cdot )=E[Z_{\varepsilon }\cdot |\mathcal{G}]$ for some $%
Z_{\varepsilon }\in L_{\mathcal{G}}^{q}(\mathcal{F}),$ from (\ref{HahnBan})
we deduce that the following holds on the set $D_{\mathcal{C}_{\varepsilon}}$: 
\begin{equation*}
E[Z_{\varepsilon }(\xi +n\eta )\mid \mathcal{G}]<E[Z_{\varepsilon }X\mid 
\mathcal{G}]\Longrightarrow E[-Z_{\varepsilon }\eta \mid \mathcal{G}]>\frac{%
E[Z_{\varepsilon }(\xi -X)\mid \mathcal{G}]}{n},\quad \forall \,n\in \mathbb{%
N}
\end{equation*}%
i.e. $E[Z_{\varepsilon }\eta \mid \mathcal{G}]\mathbf{1}_{D_{\mathcal{C}%
_{\varepsilon}}}\leq 0$ for every $\eta \in L_{\mathcal{G}}^{p}(\mathcal{F})$%
, $\eta \geq 0$. This implies, as $\mathbf{1}_{\{Z_{\varepsilon }>0\}}\in L_{%
\mathcal{G}}^{p}(\mathcal{F})$, that $Z_{\varepsilon }\mathbf{1}_{D_{%
\mathcal{C}_{\varepsilon}}}\leq 0$.

We now show that $Z_{\varepsilon }<0$ on $D_{\mathcal{C}_{\varepsilon}}$.
Suppose there existed a $\mathcal{G}$-measurable set $G\subset D_{\mathcal{C}%
_{\varepsilon}}$, $\mathbb{P}(G)>0$, on which $Z_{\varepsilon }=0$ and fix $%
\xi \in \mathcal{C}_{\varepsilon }$. From $E[Z_{\varepsilon }\xi \mid 
\mathcal{G}]<E[Z_{\varepsilon }X\mid \mathcal{G}]$ on $D_{\mathcal{C}%
_{\varepsilon}}$ we can find a $\delta _{\xi }\in L_{++}^{0}(\mathcal{G})$
such that $E[Z_{\varepsilon }\xi \mid \mathcal{G}]+\delta _{\xi
}<E[Z_{\varepsilon }X\mid \mathcal{G}]\;\text{ on }D_{\mathcal{C}%
_{\varepsilon}}$ which implies 
\begin{equation*}
\delta _{\xi}\mathbf{1}_{G}=E[Z_{\varepsilon }\mathbf{1}_{G}\xi \mid 
\mathcal{G}]+\delta _{\xi }\mathbf{1}_{G}\leq E[Z_{\varepsilon }\mathbf{1}%
_{G}X\mid\mathcal{G}]=0.
\end{equation*}
which is a contradiction since $\mathbb{P}(\delta _{\xi }\mathbf{1}_{G}>0)>0$%
. \newline
We deduce that $E[Z_{\varepsilon }\mathbf{1}_{B}]=E[E[Z_{\varepsilon }\mid 
\mathcal{G}]\mathbf{1}_{B}]<0$ for every $B\in \mathcal{G}$, $B\subseteq D_{%
\mathcal{C}_{\varepsilon}}$ and then $E[Z_{\varepsilon }\mid \mathcal{G}]<0$
on $D_{\mathcal{C}_{\varepsilon}}$. Consider any $W\in L^q_{\mathcal{G} }(%
\mathcal{F} )$ with $\mathbb{P} (W>0)=1$, $E[W\mid\mathcal{G}]=1$ and define 
$\frac{dQ_{\varepsilon }}{d\mathbb{P}}\in L^{1}(\mathcal{F})$ as 
\begin{equation*}
\frac{dQ_{\varepsilon }}{d\mathbb{P}}=\frac{Z_{\varepsilon }}{%
E[Z_{\varepsilon }\mid \mathcal{G}]}\mathbf{1}_{D_{\mathcal{C}%
_{\varepsilon}}}+W\mathbf{1}_{D^{\vdash}_{E[Z_{\varepsilon }\mid \mathcal{G}%
]}}.
\end{equation*}
\newline
Following the idea of the proof of Theorem %
\ref{EvQco} (see \cite{FM12}) we can easily deduce that 
\begin{eqnarray}
\rho (X) &\geq &\inf_{\xi \in L_{\mathcal{G}}^{p}(\mathcal{F})}\left\{ \rho
(\xi )\mid \mu _{\varepsilon }(\xi )\geq \mu _{\varepsilon }(X)\right\} 
\notag \\
&\geq&\inf_{\xi \in L_{\mathcal{G}}^{p}(\mathcal{F})}\left\{ \rho (\xi )\mid
E\left[ -\xi \frac{dQ_{\varepsilon }}{d\mathbb{P}}|\mathcal{G}\right] \geq E%
\left[ -X\frac{dQ_{\varepsilon }}{d\mathbb{P}}|\mathcal{G}\right] \right\} 
\notag \\
&\geq &(\rho (X)-\varepsilon )\mathbf{1}_{G}+\varepsilon \mathbf{1}_{G^{C}}
\label{proofDual}
\end{eqnarray}%
and hence 
\begin{equation*}
\rho (X)=\sup_{Q\in \mathcal{P}^{q}}\inf_{\xi \in L_{\mathcal{G}}^{p}(%
\mathcal{F})}\left\{ \rho (\xi )\mid E\left[ -\xi \frac{dQ}{d\mathbb{P}}|%
\mathcal{G}\right] \geq E\left[ -X\frac{dQ}{d\mathbb{P}}|\mathcal{G}\right]
\right\} .
\end{equation*}

Applying Lemma \ref{monK} we can substitute $=$ in the constraint.

\paragraph{IF.}

We assume that $\rho (X)=\sup_{Q\in \mathcal{P}^{q}}R(E[-X\frac{dQ}{d\mathbb{%
P}}\mid \mathcal{G}],Q)$ holds for some $R\in \mathcal{M}(L^{0}(\mathcal{G}%
)\times \mathcal{P}^{q})$. Since $R$ is monotone in the first component and $%
R(Y\mathbf{1}_{A},Q)\mathbf{1}_{A}=R(Y,Q)\mathbf{1}_{A}$ for every $A\in 
\mathcal{G}$ we easily deduce that $\rho $ is (MON) and (REG). We need to
show that $\rho $ is (EVQ). \newline
Let $\mathcal{V}_{\alpha }=\{\xi \in L_{\mathcal{G}}^{p}(\mathcal{F})\mid
\rho (\xi )\mathbf{1}_{T_{\rho }}\leq \alpha \}$ where $\alpha \in L^{0}(%
\mathcal{G})$ and recall that $D_{\mathcal{V}_{\alpha }}$ is the
complementary of the set provided in Definition \ref{trivialcomponent}.
Notice that $D_{\mathcal{V}_{\alpha }}\subseteq T_{\rho }$. Take $X^{\ast
}\in L_{\mathcal{G}}^{p}(\mathcal{F})$ satisfying $X^{\ast }\mathbf{1}%
_{A}\cap V_{\alpha }\mathbf{1}_{A}=\emptyset $ for every $A\in \mathcal{G}$, 
$A\subseteq D_{\mathcal{V}_{\alpha }}$, $P(A)>0$. Hence 
\begin{equation*}
\rho (X^{\ast })=\sup_{Q\in \mathcal{P}^{q}}R(E[-X^{\ast }\frac{dQ}{d\mathbb{%
P}}\mid \mathcal{G}],Q)>\alpha
\end{equation*}%
on the set $D_{\mathcal{V}_{\alpha }}$. Since the set $\{R(E[-X^{\ast }\frac{%
dQ}{d\mathbb{P}}|\mathcal{G}],Q)\mid Q\in \mathcal{P}^{q}\}$ is upward
directed there exists $\left\{ Q_{m}\right\} \subset \mathcal{P}^{q}$ s.t. 
\begin{equation*}
R\left( E\left[ -X^{\ast }\frac{dQ_{m}}{d\mathbb{P}}\mid \mathcal{G}\right]
,Q_{m}\right) \uparrow \rho (X^{\ast })\quad \text{ as }m\uparrow +\infty .
\end{equation*}%
Let $\delta \in L_{++}^{0}(\mathcal{G})$ satisfies $\delta <\rho (X)-\alpha $
and consider the sets 
\begin{equation*}
F_{m}=\left\{ R(E\left[ -X^{\ast }\frac{dQ_{m}}{d\mathbb{P}}|\mathcal{G}%
\right] ,Q_{m})>\rho (X)-\delta \right\}
\end{equation*}%
and the partition of $\Omega $ given by $G_{1}=F_{1}$ and $%
G_{m}=F_{m}\setminus G_{m-1}$. We have from the properties of the module $L_{%
\mathcal{G}}^{q}(\mathcal{F})$ that 
\begin{equation*}
\frac{dQ^{\ast }}{d\mathbb{P}}=\sum_{m=1}^{\infty }\frac{dQ_{m}}{d\mathbb{P}}%
\mathbf{1}_{G_{m}}\in L_{\mathcal{G}}^{q}(\mathcal{F})
\end{equation*}%
and then $Q^{\ast }\in \mathcal{P}^{q}$ with $R(E[-X^{\ast }\frac{dQ^{\ast }%
}{d\mathbb{P}}|\mathcal{G}],Q^{\ast })>\alpha $ on the set $D_{\mathcal{V}%
_{\alpha }}$.

Let $\xi \in V_{\alpha }$. It remains to show that this $Q^{\ast }$
separates $X^{\ast }$ from $V_{\alpha }$ on the set $D_{\mathcal{V}%
_{\alpha}} $. If there existed $A\subseteq D_{\mathcal{V}_{\alpha}}\in 
\mathcal{G}$ such that $E[\xi \frac{dQ^{\ast }}{d\mathbb{P}}\mathbf{1}_{A}|%
\mathcal{G}]\leq E[X^{\ast }\frac{dQ^{\ast }}{d\mathbb{P}}\mathbf{1}_{A}|%
\mathcal{G}]$ on $A$ then $\rho (\xi \mathbf{1}_{A})\geq R(E[-\xi \frac{%
dQ^{\ast }}{d\mathbb{P}}\mathbf{1}_{A}|\mathcal{G}],Q^{\ast })\geq
R(E[-X^{\ast }\frac{dQ^{\ast }}{d\mathbb{P}}\mathbf{1}_{A}|\mathcal{G}%
],Q^{\ast })>\alpha $ on $A $. This implies $\rho (\xi )>\alpha $ on $A$
which is a contradiction unless $\mathbb{P}(A)=0$. Hence $E[\xi \frac{%
dQ^{\ast }}{d\mathbb{P}}|\mathcal{G}]>E[X^{\ast }\frac{dQ^{\ast }}{d\mathbb{P%
}}|\mathcal{G}]$ on $D_{\mathcal{V}_{\alpha}}$ for every $\xi \in V_{\alpha
} $.

\paragraph{UNIQUENESS.}

First we need the following preliminary result. Define the set 
\begin{equation*}
\mathcal{A}(Y,Q)=\left\{ \xi \in L_{\mathcal{G}}^{p}(\mathcal{F})\mid E\left[
-\xi \frac{dQ}{d\mathbb{P}}|\mathcal{G}\right] \geq Y\right\} .
\end{equation*}

\begin{lemma}
\label{program}If $K\in \mathcal{M}(L^{0}(\mathcal{G})\times \mathcal{P}%
^{q}) $, then for each $(Y^{\ast },Q^{\ast })\in L^{0}(\mathcal{G})\times 
\mathcal{P}^{q}$ 
\begin{equation}
K(Y^{\ast },Q^{\ast })=\sup_{Q\in \mathcal{P}^{q}}\inf_{X\in \mathcal{A}%
(Y^{\ast },Q^{\ast })}K\left( E\left[ -X\frac{dQ}{d\mathbb{P}}|\mathcal{G}%
\right] ,Q\right)  \label{claim}
\end{equation}
\end{lemma}

\begin{proof}
Consider 
\begin{equation*}
\psi (Q,Q^{\ast },Y^{\ast })=\inf_{X\in \mathcal{A}(Y^{\ast },Q^{\ast
})}K\left( E\left[ -X\frac{dQ}{d\mathbb{P}}|\mathcal{G}\right] ,Q\right)
\end{equation*}%
Notice that $E[-X\frac{dQ^{\ast }}{d\mathbb{P}}|\mathcal{G}]\geq Y^{\ast }$
for every $X\in \mathcal{A}(Y^{\ast },Q^{\ast })$ implies 
\begin{equation*}
\psi (Q^{\ast },Q^{\ast },Y^{\ast })=\inf_{X\in \mathcal{A}(Y^{\ast
},Q^{\ast })}K\left( E\left[ -X\frac{dQ^{\ast }}{d\mathbb{P}}|\mathcal{G}%
\right] ,Q^{\ast }\right) \geq K(Y^{\ast },Q^{\ast })
\end{equation*}%
On the other hand $E[Y^{\ast }\frac{dQ^{\ast }}{d\mathbb{P}}|\mathcal{G}%
]=Y^{\ast }$ so that $-Y^{\ast }\in \mathcal{A}(Y^{\ast },Q^{\ast })$ and
the second inequality is actually an equality 
\begin{equation*}
\psi (Q^{\ast },Q^{\ast },Y^{\ast })\leq K\left( E\left[ -(-Y^{\ast })\frac{%
dQ^{\ast }}{d\mathbb{P}}|\mathcal{G}\right] ,Q^{\ast }\right) =K(Y^{\ast
},Q^{\ast }).
\end{equation*}%
\newline
If we show that $\psi (Q,Q^{\ast },Y^{\ast })\leq \psi (Q^{\ast },Q^{\ast
},Y^{\ast })$ for every $Q\in \mathcal{P}^{q}$ then (\ref{claim}) is proved.
To this aim we define%
\begin{equation*}
\mathcal{A}:=\left\{ A\in \mathcal{G}\mid E\left[ X\frac{dQ^{\ast }}{d%
\mathbb{P}}|\mathcal{G}\right] 1_{A}=E\left[ X\frac{dQ}{d\mathbb{P}}|%
\mathcal{G}\right] 1_{A},\;\text{ }\forall \,X\in L_{\mathcal{G}}^{p}(%
\mathcal{F})\right\}
\end{equation*}%
For every $A\in \mathcal{A}$ and every $X\in L_{\mathcal{G}}^{p}(\mathcal{F}%
) $ 
\begin{eqnarray*}
K\left( E\left[ -X\frac{dQ}{d\mathbb{P}}|\mathcal{G}\right] ,Q\right) 
\mathbf{1}_{A} &=&K\left( E\left[ -X\frac{dQ}{d\mathbb{P}}|\mathcal{G}\right]
\mathbf{1}_{A},Q\right) \mathbf{1}_{A} \\
=K\left( E\left[ -X\frac{dQ^{\ast }}{d\mathbb{P}}|\mathcal{G}\right] \mathbf{%
1}_{A},Q^{\ast }\right) \mathbf{1}_{A} &=&K\left( E\left[ -X\frac{dQ^{\ast }%
}{d\mathbb{P}}|\mathcal{G}\right] ,Q^{\ast }\right) \mathbf{1}_{A}
\end{eqnarray*}%
which implies 
\begin{equation}
\psi (Q,Q^{\ast },Y^{\ast })\mathbf{1}_{A}=\psi (Q^{\ast },Q^{\ast },Y^{\ast
})\mathbf{1}_{A}.  \label{7878}
\end{equation}%
Notice that $\mathcal{A}=\left\{ A\in \mathcal{G}\mid Y=0\;\text{on }A\text{%
, }\forall Y\in F\right\} $, where%
\begin{equation*}
F:=\left\{ E\left[ X\frac{dQ^{\ast }}{d\mathbb{P}}|\mathcal{G}\right] -E%
\left[ X\frac{dQ}{d\mathbb{P}}|\mathcal{G}\right] \mid X\in L_{\mathcal{G}%
}^{p}(\mathcal{F})\right\} .
\end{equation*}%
As the conditional expectation is (REG), we may apply Lemma \ref{L10} and
deduce the existence of two maximal sets $A_{M}\in \mathcal{A}$ and $%
A_{M}^{\vdash }\in \mathcal{A}^{\vdash }\mathcal{\ }$such that: $P(A_{M}\cap
A_{M}^{\vdash })=0,$ $P(A_{M}\cup A_{M}^{\vdash })=1$; $E\left[ X\frac{%
dQ^{\ast }}{d\mathbb{P}}|\mathcal{G}\right] \mathbf{1}_{A_{M}}=E\left[ X%
\frac{dQ}{d\mathbb{P}}|\mathcal{G}\right] \mathbf{1}_{A_{M}},\;\forall
\,X\in L_{\mathcal{G}}^{p}(\mathcal{F});$ and $E\left[ -X^{\ast }\frac{%
dQ^{\ast }}{d\mathbb{P}}|\mathcal{G}\right] \neq E\left[ -X^{\ast }\frac{dQ}{%
d\mathbb{P}}|\mathcal{G}\right] $ on $A_{M}^{\vdash }$, for some $X^{\ast
}\in L_{\mathcal{G}}^{p}(\mathcal{F}).$ Considering $A_{M}\in \mathcal{A}$
we then deduce from (\ref{7878})%
\begin{equation*}
\psi (Q,Q^{\ast },Y^{\ast })\mathbf{1}_{A_{M}}=\psi (Q^{\ast },Q^{\ast
},Y^{\ast })\mathbf{1}_{A_{M}}.
\end{equation*}%
Now we consider $A_{M}^{\vdash }\in \mathcal{A}^{\vdash }$ and define $%
Z^{\ast }:=X^{\ast }-E\left[ -X^{\ast }\frac{dQ^{\ast }}{d\mathbb{P}}|%
\mathcal{G}\right] $. Surely $E\left[ Z^{\ast }\frac{dQ^{\ast }}{d\mathbb{P}}%
|\mathcal{G}\right] =0$ and $E\left[ Z^{\ast }\frac{dQ}{d\mathbb{P}}|%
\mathcal{G}\right] \neq 0$ on $A_{M}^{\vdash }$. We deduce that for every $%
\alpha \in L^{0}(\mathcal{G})$, $-Y^{\ast }+\alpha Z^{\ast }\in \mathcal{A}%
(Y^{\ast },Q^{\ast }).$ Also notice that any $Y\in L^{0}(\mathcal{G})$ can
be written as $Y\mathbf{1}_{A_{M}^{\vdash }}=E[(-Y^{\ast }+\alpha
_{Y}Z^{\ast })\frac{dQ}{d\mathbb{P}}|\mathcal{G}]\mathbf{1}_{A_{M}^{\vdash
}} $, with $\alpha _{Y}\in L^{0}(\mathcal{G})$. Finally 
\begin{eqnarray*}
\psi (Q,Q^{\ast },Y^{\ast })\mathbf{1}_{A_{M}^{\vdash }} &\leq &\inf_{\alpha
\in L^{0}(\mathcal{G})}K\left( E\left[ -(-Y^{\ast }+\alpha Z^{\ast })\frac{dQ%
}{d\mathbb{P}}|\mathcal{G}\right] ,Q\right) \mathbf{1}_{A_{M}^{\vdash }} \\
&=&\inf_{Y\in L^{0}(\mathcal{G})}K\left( Y\mathbf{1}_{A_{M}^{\vdash
}},Q\right) \mathbf{1}_{A_{M}^{\vdash }}=\inf_{Y\in L^{0}(\mathcal{G}%
)}K\left( Y\mathbf{1}_{A_{M}^{\vdash }},Q^{\ast }\right) \mathbf{1}%
_{A_{M}^{\vdash }} \\
&\leq &K\left( Y^{\ast },Q^{\ast }\right) \mathbf{1}_{A_{M}^{\vdash }}.
\end{eqnarray*}%
As $P(A_{M}\cup A_{M}^{\vdash })=1,$ we conclude that $\psi (Q,Q^{\ast
},Y^{\ast })\leq \psi (Q^{\ast },Q^{\ast },Y^{\ast })=K(Y^{\ast },Q^{\ast })$
and the claim is proved.
\end{proof}

\bigskip

To prove the uniqueness we show that for every $K\in \mathcal{M}(L^{0}(%
\mathcal{G})\times \mathcal{P}^{q})$ such that 
\begin{equation*}
\rho (X)=\sup_{Q\in \mathcal{P}^{q}}K\left( E\left[ -X\frac{dQ}{d\mathbb{P}}|%
\mathcal{G}\right] ,Q\right) ,
\end{equation*}%
$K$ must satisfy 
\begin{equation*}
K(Y,Q)=\inf_{\xi \in L_{\mathcal{G}}^{p}(\mathcal{F})}\left\{ \rho (\xi
)\mid E\left[ -\xi \frac{dQ}{d\mathbb{P}}|\mathcal{G}\right] \geq Y\right\} .
\end{equation*}%
By the Lemma \ref{program}%
\begin{eqnarray*}
K(Y^{\ast },Q^{\ast }) &=&\sup_{Q\in \mathcal{P}^{q}}\inf_{X\in \mathcal{A}%
(Y^{\ast },Q^{\ast })}K\left( E\left[ -X\frac{dQ}{d\mathbb{P}}|\mathcal{G}%
\right] ,Q\right) \\
&\leq &\inf_{X\in \mathcal{A}(Y^{\ast },Q^{\ast })}\sup_{Q\in \mathcal{P}%
^{q}}K\left( E\left[ -X\frac{dQ}{d\mathbb{P}}|\mathcal{G}\right] ,Q\right)
=\inf_{X\in \mathcal{A}(Y^{\ast },Q^{\ast })}\rho (X).
\end{eqnarray*}%
It remains to prove the reverse inequality, i.e. 
\begin{equation}
K(Y^{\ast },Q^{\ast })\leq \inf_{X\in \mathcal{A}(Y^{\ast },Q^{\ast })}\rho
(X).  \label{1234}
\end{equation}%
Consider the class:%
\begin{equation*}
\mathcal{A}:=\left\{ A\in \mathcal{G}\mid K(Y,Q)\mathbf{1}_{A}\leq K(Y^{\ast
},Q^{\ast })\mathbf{1}_{A}\;\forall \,(Y,Q)\in L^{0}(\mathcal{G})\times 
\mathcal{P}^{q}\right\} .
\end{equation*}%
Notice that $\mathcal{A}=\left\{ A\in \mathcal{G}\mid Z\leq Z_{0}\;\text{on }%
A\text{, }\forall Z\in F\right\} $, where%
\begin{equation*}
F:=\left\{ K(Y,Q)\mid (Y,Q)\in L^{0}(\mathcal{G})\times \mathcal{P}%
^{q}\right\} .
\end{equation*}%
and $Z_{0}=K(Y^{\ast },Q^{\ast }).$ In order to apply Lemma \ref{L10}, let $%
A_{i}\in \mathcal{A}^{\vdash }$ be a sequence of disjoint sets and $%
Z_{i}=K(Y_{i},Q_{i})$ be the corresponding element in $F$. From $K(Y\mathbf{1%
}_{A},Q)\mathbf{1}_{A}=K(Y,Q)\mathbf{1}_{A}$ we deduce (as in Remark \ref%
{remA} i) that $K(\sum\limits_{i=1}^{\infty }Y_{i}\mathbf{1}_{A_{i}},Q_{j})%
\mathbf{1}_{A_{\infty }}=\sum\limits_{i=1}^{\infty }K(Y_{i},Q_{j})\mathbf{1}%
_{A_{i}}$, with $A_{\infty }=\cup A_{i}$. From $K(Y,Q_{1})\mathbf{1}%
_{A}=K(Y,Q_{2})\mathbf{1}_{A}$ if $\frac{dQ_{1}}{d\mathbb{P}}\mathbf{1}_{A}=%
\frac{dQ_{2}}{d\mathbb{P}}\mathbf{1}_{A}$, $A\in \mathcal{G}$, we obtain 
\begin{eqnarray*}
K(\sum\limits_{i=1}^{\infty }Y_{i}\mathbf{1}_{A_{i}},\sum\limits_{j=1}^{%
\infty }Q_{j}\mathbf{1}_{A_{j}})\mathbf{1}_{A_{\infty }}
&=&\sum\limits_{i=1}^{\infty }K(Y_{i},\sum\limits_{j=1}^{\infty }Q_{j}%
\mathbf{1}_{A_{j}})\mathbf{1}_{A_{i}} \\
&=&\sum\limits_{i=1}^{\infty }K(Y_{i},Q_{i})\mathbf{1}_{A_{i}}=\sum%
\limits_{i=1}^{\infty }Z_{i}\mathbf{1}_{A_{i}}
\end{eqnarray*}%
showing that $\sum\limits_{i=1}^{\infty }Z_{i}\mathbf{1}_{A_{i}}\in F.$ From
Lemma \ref{L10} we may deduce the existence of two maximal sets $A_{M}\in 
\mathcal{A}$ and $A_{M}^{\vdash }\in \mathcal{A}^{\vdash }\mathcal{\ }$such
that: $P(A_{M}\cap A_{M}^{\vdash })=0,$ $P(A_{M}\cup A_{M}^{\vdash })=1$; $%
K(Y,Q)\mathbf{1}_{A_M}\leq K(Y^{\ast },Q^{\ast })\mathbf{1}_{A_M}\;\forall
\,(Y,Q)\in L^{0}(\mathcal{G})\times \mathcal{P}^{q};$ and 
\begin{equation}
K(Y^{\ast },Q^{\ast })<K(\overline{Y},\overline{Q})\text{ on }A_{M}^{\vdash
},  \label{2345}
\end{equation}%
for some $(\overline{Y},\overline{Q})\in L^{0}(\mathcal{G})\times \mathcal{P}%
^{q}.$ On $A_{M}\in \mathcal{A}$ the inequality (\ref{1234}) is obviously
true and we need only to show it on the set $A_{M}^{\vdash }$.

From (\ref{2345}) we can easily build a $\beta \in L^{0}(\mathcal{G})$ such
that $K(Y^{\ast },Q^{\ast })<\beta \leq K(\overline{Y},\overline{Q})$ on $%
A_{M}^{\vdash }$ and $\beta $ is arbitrarily close to $K(Y^{\ast },Q^{\ast
}) $ on $A_{M}^{\vdash }$. An example of such $\beta $ is obtained by taking 
$\lambda \downarrow 0$ in the family:%
\begin{eqnarray*}
\beta _{\lambda } &:=&1_{A_{M}^{\vdash }} \left[ \lambda K(\overline{Y},%
\overline{Q})+(1-\lambda )K(Y^{\ast },Q^{\ast })\right] 1_{\left\{ K(%
\overline{Y},\overline{Q})<\infty \right\} \cap \left\{ K(Y^{\ast },Q^{\ast
})>-\infty \right\} } \\
&+&1_{A_{M}^{\vdash }}1_{\left\{ K(\overline{Y},\overline{Q})=\infty
\right\} }\left[ (K(Y^{\ast },Q^{\ast })+\lambda )1_{\left\{ K(Y^{\ast
},Q^{\ast })>-\infty \right\} }-\frac{1}{\lambda }1_{\left\{ K(Y^{\ast
},Q^{\ast })=-\infty \right\} }\right] .
\end{eqnarray*}%
Since the set $U_{\beta }:=\{(Y,Q)\in L^{0}(\mathcal{G})\times \mathcal{P}%
^{q}\mid K(Y,Q)\geq \beta \text{ on }A_{M}^{\vdash }\}$ is not empty, the
assumption that $K$ is $\diamond $-evenly $L^{0}(\mathcal{G})$-quasiconcave
implies the existence of $(S^{\ast },X^{\ast })\in L_{++}^{0}(\mathcal{G}%
)\times L_{\mathcal{G}}^{p}(\mathcal{F})$ with 
\begin{equation*}
Y^{\ast }S^{\ast }+E\left[ X^{\ast }\frac{dQ^{\ast }}{d\mathbb{P}}\mid 
\mathcal{G}\right] <YS^{\ast }+E\left[ X^{\ast }\frac{dQ}{d\mathbb{P}}\mid 
\mathcal{G}\right] \text{ on }A_{M}^{\vdash }
\end{equation*}%
for every $(Y,Q)\in U_{\beta }$.

We claim that for every $(Y,Q)\in U_{\beta }$ 
\begin{equation*}
Y+E\left[ \widehat{X}\frac{dQ}{d\mathbb{P}}\mid \mathcal{G}\right] >0\text{
on }A_{M}^{\vdash }\text{ ,}
\end{equation*}%
where $\widehat{X}:=\frac{X^{\ast }}{S^{\ast }}+\Lambda $ and $\Lambda
:=-Y^{\ast }-E[\frac{X^{\ast }}{S^{\ast }}\frac{dQ^{\ast }}{d\mathbb{P}}\mid 
\mathcal{G}]$. Indeed, for every $(Y,Q)\in U_{\beta }$ 
\begin{eqnarray*}
&&Y^{\ast }S^{\ast }+E\left[ X^{\ast }\frac{dQ^{\ast }}{d\mathbb{P}}\mid 
\mathcal{G}\right] <YS^{\ast }+E\left[ X^{\ast }\frac{dQ}{d\mathbb{P}}\mid 
\mathcal{G}\right] \text{ on }A_{M}^{\vdash }\text{ ,} \\
\text{implies }\; &&Y^{\ast }+E\left[ \left( \frac{X^{\ast }}{S^{\ast }}%
+\Lambda \right) \frac{dQ^{\ast }}{d\mathbb{P}}\mid \mathcal{G}\right] <Y+E%
\left[ \left( \frac{X^{\ast }}{S^{\ast }}+\Lambda \right) \frac{dQ}{d\mathbb{%
P}}\mid \mathcal{G}\right] \text{ on }A_{M}^{\vdash }\text{ ,} \\
\text{implies }\; &&Y^{\ast }+E\left[ \widehat{X}\frac{dQ^{\ast }}{d\mathbb{P%
}}\mid \mathcal{G}\right] <Y+E\left[ \widehat{X}\frac{dQ}{d\mathbb{P}}\mid 
\mathcal{G}\right] \text{ on }A_{M}^{\vdash }\text{ },
\end{eqnarray*}%
i.e. the claim holds, as $E[\widehat{X}\frac{dQ^{\ast }}{d\mathbb{P}}\mid 
\mathcal{G}]=-Y^{\ast }$.

For every $Q\in \mathcal{P}^{q}$ define $Y_{Q}:=E\left[ -\widehat{X}\frac{dQ%
}{d\mathbb{P}}\mid \mathcal{G}\right] $. We show that 
\begin{equation}
K(Y_{Q},Q)<\beta \text{ on }A_{M}^{\vdash }.  \label{2424}
\end{equation}%
Suppose by contradiction that there exists $B\subseteq A_{M}^{\vdash }$, $%
B\in \mathcal{G}$, $P(B)>0$, such that $K(Y_{Q},Q)\geq \beta $ on $B$. Take $%
(Y_{1},Q_{1})\in U_{\beta }$ and define $\widetilde{Y}:=Y_{Q}\mathbf{1}%
_{B}+Y_{1}\mathbf{1}_{B^{C}}$ and $\widetilde{Q}\in \mathcal{P}^{q}$ by 
\begin{equation*}
\frac{d\tilde{Q}}{d\mathbb{P}}=\frac{dQ}{d\mathbb{P}}\mathbf{1}_{B}+\frac{%
dQ_{1}}{d\mathbb{P}}\mathbf{1}_{B^{C}}.
\end{equation*}%
Thus $K(\widetilde{Y},\widetilde{Q})\geq \beta $ on $A_{M}^{\vdash }$ and $%
\widetilde{Y}+E\left[ \widehat{X}\frac{d\tilde{Q}}{d\mathbb{P}}\mid \mathcal{%
G}\right] >0$ on $A_{M}^{\vdash }$, which implies $Y_{Q}+E\left[ \widehat{X}%
\frac{dQ}{d\mathbb{P}}\mid \mathcal{G}\right] >0$ on $B$ and this is
impossible and (\ref{2424}) is proven.

Since $\widehat{X}\in \mathcal{A}(Y^{\ast },Q^{\ast })$ we can conclude that 
\begin{eqnarray*}
K(Y^{\ast },Q^{\ast })\mathbf{1}_{A_{M}^{\vdash }} &\leq &\inf_{X\in 
\mathcal{A}(Y^{\ast },Q^{\ast })}\sup_{Q\in \mathcal{P}^{q}}K\left( E\left[
-X\frac{dQ}{d\mathbb{P}}\mid \mathcal{G}\right] ,Q\right) \mathbf{1}%
_{A_{M}^{\vdash }} \\
&\leq &\sup_{Q\in \mathcal{P}^{q}}K\left( E\left[ -\widehat{X}\frac{dQ}{d%
\mathbb{P}}\mid \mathcal{G}\right] ,Q\right) \mathbf{1}_{A_{M}^{\vdash
}}\leq \beta \mathbf{1}_{A_{M}^{\vdash }}.
\end{eqnarray*}%
As $\beta $ is arbitrarily close to $K(Y^{\ast },Q^{\ast }),$ the equality
must hold and then we obtain: 
\begin{equation*}
K(Y^{\ast },Q^{\ast })=\inf_{X\in \mathcal{A}(Y^{\ast },Q^{\ast })}\rho (X)%
\text{ on }A_{M}^{\vdash }.
\end{equation*}%
This concludes the proof of Theorem \ref{RM}.

\section{Appendix}

\begin{remark}
\label{remMAX}By Lemma 2.9 in \cite{FKV09}, we know that any non-empty class 
$\mathcal{A}$ of subsets of a sigma algebra $\mathcal{G}$ has a supremum $%
\emph{ess}.\sup \{\mathcal{A}\}\in \mathcal{G}$ and that if $\mathcal{A}$ is
closed with respect to finite union (i.e. $A_{1},A_{2}\in \mathcal{A}%
\Rightarrow A_{1}\cup A_{2}\in \mathcal{A}$) then there is a sequence $%
A_{n}\in \mathcal{A}$ such that $\emph{ess}.\sup \{\mathcal{A}%
\}=\bigcup\limits_{n\in \mathbb{N}}A_{n}$. Obviously, if $\mathcal{A}$ is
closed with respect to countable union then $\emph{ess}.\sup \{\mathcal{A}%
\}=\bigcup\limits_{n\in \mathbb{N}}A_{n}:=A_{M}\in \mathcal{A}$ is the
maximal element in $\mathcal{A}$.
\end{remark}

The next Lemma is used several times in the proofs of the paper. It says
that for any subset $F\subset L^{0}(\mathcal{G})$ that is \textquotedblleft
closed w.r.to pasting\textquotedblright\ it is possible to determine a
maximal set $A_{M}\in \mathcal{G}$ (which may have zero probability) such
that $Y1_{A_{M}}\geq 0$ $\forall Y\in F$ and one element $\overline{Y}\in F$
for which $\overline{Y}<0$ on the complement of $A_{M}$.

\begin{lemma}
\label{L10}With the symbol $\trianglerighteq $ denote any one of the binary
relations $\geq ,$ $\leq ,$ $=,$ $>$, $<$ and with $\vartriangleleft $ its
negation. Consider a class $F\subseteq \bar{L}^{0}(\mathcal{G})$ of random
variables, $Y_{0}\in \bar{L}^{0}(\mathcal{G})$ and the classes of sets 
\begin{eqnarray*}
\mathcal{A}:= &\{A\in &\mathcal{G}\mid \forall \,Y\in F\text{ }%
Y\trianglerighteq Y_{0}\text{ on }A\}, \\
\mathcal{A}^{\vdash }:= &\{A^{\vdash }\in &\mathcal{G}\mid \exists \,Y\in F%
\text{ s.t. }Y\vartriangleleft Y_{0}\text{ on }A^{\vdash }\}.
\end{eqnarray*}%
Suppose that for any sequence of disjoint sets $A_{i}^{\vdash }\in \mathcal{A%
}^{\vdash }$ and the associated r.v. $Y_{i}\in F$ we have $\sum_{1}^{\infty
}Y_{i}1_{A_{i}^{\vdash }}\in F$. Then there exist two maximal sets $A_{M}\in 
\mathcal{A}$ and $A_{M}^{\vdash }\in \mathcal{A}^{\vdash }\mathcal{\ }$such
that $P(A_{M}\cap A_{M}^{\vdash })=0$, $P(A_{M}\cup A_{M}^{\vdash })=1$ and 
\begin{eqnarray*}
&&Y\trianglerighteq Y_{0}\text{ on }A_{M}\text{, }\forall Y\in F\text{ } \\
&&\overline{Y}\vartriangleleft Y_{0}\text{ on }A_{M}^{\vdash }\text{, for
some }\overline{Y}\in F.
\end{eqnarray*}
\end{lemma}

\begin{proof}
Notice that $\mathcal{A}$ and $\mathcal{A}^{\vdash }$ are closed with
respect to \emph{countable} union. This claim is obvious for $\mathcal{A}$.
For $\mathcal{A}^{\vdash }$, suppose that $A_{i}^{\vdash }\in \mathcal{A}%
^{\vdash }$ and that $Y_{i}\in F$ satisfies $P(\left\{ Y_{i}\vartriangleleft
Y_{0}\right\} \cap A_{i}^{\vdash })=P(A_{i}^{\vdash }).$ Defining $%
B_{1}:=A_{i}^{\vdash }$, $B_{i}:=A_{i}^{\vdash }\setminus B_{i-1}$, $%
A_{\infty }^{\vdash }:=\bigcup\limits_{i=1}^{\infty }A_{i}^{\vdash
}=\bigcup\limits_{i=1}^{\infty }B_{i}$ we see that $B_{i}$ are disjoint
elements of $\mathcal{A}^{\vdash }$ and that $Y^{\ast }:=\sum_{1}^{\infty
}Y_{i}1_{B_{i}}\in F$ satisfies $P(\left\{ Y^{\ast }\vartriangleleft
Y_{0}\right\} \cap A_{\infty }^{\vdash })=P(A_{\infty }^{\vdash })$ and so $%
A_{\infty }^{\vdash }\in \mathcal{A}^{\vdash }$.

The Remark \ref{remMAX} guarantees the existence of two sets $A_{M}\in 
\mathcal{A}$ and $A_{M}^{\vdash }\in \mathcal{A}^{\vdash }$ such that:

(a) $P(A\cap (A_{M})^{C})=0$ for all $A\in \mathcal{A}$,

(b) $P(A^{\vdash }\cap (A_{M}^{\vdash })^{C})=0$ for all $A^{\vdash }\in 
\mathcal{A}^{\vdash }$.

Obviously, $P(A_{M}\cap A_{M}^{\vdash })=0,$ as $A_{M}\in \mathcal{A}$ and $%
A_{M}^{\vdash }\in \mathcal{A}^{\vdash }$. To show that $P(A_{M}\cup
A_{M}^{\vdash })=1$, let $D:=\Omega \setminus \left\{ A_{M}\cup
A_{M}^{\vdash }\right\} \in \mathcal{G}$. By contradiction suppose that $%
P(D)>0$. As $D\subseteq (A_{M})^{C}$, from condition (a) we get $D\notin 
\mathcal{A}$. Therefore, $\exists \overline{Y}\in F$ s.t. $P(\left\{ 
\overline{Y}\trianglerighteq Y_{0}\right\} \cap D)<P(D),$ i.e. $P(\left\{ 
\overline{Y}\vartriangleleft Y_{0}\right\} \cap D)>0$. If we set $B:=\left\{ 
\overline{Y}\vartriangleleft Y_{0}\right\} \cap D$ then it satisfies $%
P(\left\{ \overline{Y}\vartriangleleft Y_{0}\right\} \cap B)=P(B)>0$ and, by
definition of $\mathcal{A}^{\vdash }$, $B$ belongs to $\mathcal{A}^{\vdash }$%
. On the other hand, as $B\subseteq D\subseteq (A_{M}^{\vdash })^{C}$, $%
P(B)=P(B\cap (A_{M}^{\vdash })^{C}),$ and from condition (b) $P(B\cap
(A_{M}^{\vdash })^{C})=0$, which contradicts $P(B)>0$.
\end{proof}


\begin{thebibliography}{FKV10}
\bibitem[CV09]{CMMMb} \textsc{Cerreia-Vioglio, S., Maccheroni, F.,
Marinacci, M. and Montrucchio, L.} (2009) \textquotedblleft Complete
Monotone Quasiconcave Duality \textquotedblright , forthcoming on \textit{Math. Op. Res.}

\bibitem[CV10]{CMMMa} \textsc{Cerreia-Vioglio, S., Maccheroni, F.,
Marinacci, M. and Montrucchio, L.} (2010) \textquotedblleft Risk measures:
rationality and diversification\textquotedblright , forthcoming on \textit{Math. Fin.}.

\bibitem[CM09]{CM} \textsc{Cherny, A. and Madan, D.} (2009) \textquotedblleft New measures for performance evaluation\textquotedblright, \textit{Review of Financial Studies}, \textbf{22}, 2571-2606.

\bibitem[DS05]{Sca} \textsc{Detlefsen, K. and Scandolo, G.} (2005)
\textquotedblleft Conditional and dynamic convex risk
measures\textquotedblright , \textit{Finance and Stochastics}, \textbf{9}, 539-561 .


\bibitem[DK10]{KD} \textsc{Drapeau, S. and Kupper, M.} (2010) Risk
preferences and their robust representation, forthcoming on \textit{ Math. Op. Research}.

\bibitem[ER09]{ER09} \textsc{El Karoui, N. and Ravanelli, C.} (2009)
\textquotedblleft Cash sub-additive risk measures and interest rate
ambiguity\textquotedblright , \textit{Mathematical Finance}, \textbf{19}%
(4), 561-590.

\bibitem[FKV10]{FKV10} \textsc{Filipovic, D. Kupper, M. and Vogelpoth, N.}
(2010) \textquotedblleft Approaches to conditional risk \textquotedblright , 
forthcoming on  \textit{SIAM J. Fin. Math.}.

\bibitem[FKV09]{FKV09} \textsc{Filipovic, D. Kupper, M. and Vogelpoth, N.}
(2009) \textquotedblleft Separation and duality in locally $L^{0}$-convex
modules \textquotedblright , \textit{Journal of Functional Analysis },  
\textbf{256}(12), 3996-4029.

\bibitem[FP06]{FoPe} \textsc{F$\ddot{o}$llmer, H. and Penner, I.} (2006)
\textquotedblleft Convex risk measures and the dynamics of their penalty
functions \textquotedblright , \textit{Statistics and Decisions },
\textbf{24}(1), 61-96.

\bibitem[FS02]{FoSchA} \textsc{F$\ddot{o}$llmer, H. and Shied, A.} (2002)
\textquotedblleft Convex measures of risk and trading constraints
\textquotedblright , \textit{Finance and Stochastics}, \textbf{\ 6}, 429-447.

\bibitem[FM12]{FM12} \textsc{Frittelli, M. and Maggis, M.} (2012)
\textquotedblleft Conditionally evenly convex sets and evenly quasi-convex maps\textquotedblright , \textit{ArXiv},.

\bibitem[FM11a]{FM10} \textsc{Frittelli, M. and Maggis, M.} (2011)
\textquotedblleft Conditional Certainty Equivalent \textquotedblright , 
\textit{Int. J. Theor. Appl. Fin.}, \textbf{14}(1), 41-59.

\bibitem[FM11b]{FM09} \textsc{Frittelli, M. and Maggis, M.} (2011)
\textquotedblleft Dual representation of quasiconvex conditional
maps\textquotedblright , \textit{SIAM J. Fin. Math.}, \textbf{2}, 357-382.

\bibitem[FMP12]{FMP12} \textsc{Frittelli, M., Maggis, M. and Peri, I.} (2012)
\textquotedblleft Risk Measures on $\mathcal{P}(\R)$ and Value At Risk with
Probability/Loss function \textquotedblright , \textit{ArXiv}, 1201.2257v3.


\bibitem[FR02]{FrM} \textsc{Frittelli, M. and Rosazza Gianin, E.} (2002)
\textquotedblleft Putting order in risk measures \textquotedblright , 
\textit{Journal of Banking and Finance},  \textbf{26}(7), 1473-1486.

\bibitem[FR04]{FR1} \textsc{Frittelli, M., Rosazza Gianin, E.} (2004).
Dynamic Convex Risk Measures, in: Risk Measures for the 21st Century, G. Szeg\"{o} ed., J. Wiley, pp. 227-248.

\bibitem[Gu10]{Guo} \textsc{Guo, T.X.}(2010) \textquotedblleft Relations
between some basic results derived from two kinds of topologies for a random locally convex module \textquotedblright , \textit{Journal of Functional Analysis }, \textbf{258}, 3024-3047.

\bibitem[Gu11]{Guo1} \textsc{Guo, T.X.}(2011) \textquotedblleft Recent
progress in random metric theory and its applications to conditional risk
measures \textquotedblright , \textit{ArXiv}, 1006.0697v17.


\bibitem[KS10]{KS10} \textsc{Kupper, M. and Schachermayer, W.} (2009)
\textquotedblleft Representation Results for Law Invariant Time Consistent
Functions \textquotedblright , \textit{Mathematics and Financial Economics,}  textbf{2}(3), 189-2101.

\bibitem[KV09]{KV09} \textsc{Kupper, M. and Vogelpoth, N.} (2009)
\textquotedblleft Complete $L^{0}$-modules and automatic continuity of
monotone convex functions \textquotedblright , \textit{VIF Working Papers
Series}.


\bibitem[PV90]{PV} \textsc{Penot, J.P. and Volle, M.} (1990)
\textquotedblleft On Quasi-Convex Duality\textquotedblright , \textit{%
Mathematics of Operations Research}, \textbf{15}(4), 597-625.

\end{thebibliography}
\end{document}